%% file: main.tex
\documentclass[opre,nonblindrev]{informs4rad}
\OneAndAHalfSpacedXI

\input{Paper/macros}

\input{Paper/math}

\MANUSCRIPTNO{}

\begin{document}


\RUNAUTHOR{Feng, Niazadeh, Saberi}

\newcommand{\revcolor}[1]{{\color{black}#1}}

\RUNTITLE{Robustness of Online Inventory Balancing to Inventory Shocks}
\TITLE{Robustness of Online Inventory Balancing to Inventory Shocks}

\ARTICLEAUTHORS{%
\AUTHOR{Yiding Feng}
\AFF{Industrial Engineering and Decision Analytics, Hong Kong University of Science \& Technology (HKUST), \EMAIL{ydfeng@ust.hk}}
\AUTHOR{Rad Niazadeh}
\AFF{The University of Chicago Booth School of Business, \EMAIL{rad.niazadeh@chicagobooth.edu}}
\AUTHOR{Amin Saberi}
\AFF{Management Science and Engineering, Stanford University, \EMAIL{saberi@stanford.edu}}
} 

\ABSTRACT{%
\input{Paper/abstract}
}

\maketitle
\setcounter{page}{1}
\newpage

\section{Introduction}
\input{Paper/intro}

\section{Preliminaries}
\input{Paper/prelim}

\section{Asymptotically Optimal Algorithm: \IBB}
\input{Paper/main-result}

\section{Competitive Ratio Analysis: Proof of Theorem~\ref{thm:concave competitive ratio}}
\input{Paper/cr-analysis}

\input{Paper/conclusion}

\newcommand{\newblock}{}
\setlength{\bibsep}{0.0pt}
\bibliographystyle{plainnat}
\OneAndAHalfSpacedXI
{\footnotesize

\newpage
\bibliography{refs}}


\newpage
\renewcommand{\theHchapter}{A\arabic{chapter}}
\renewcommand{\theHsection}{A\arabic{section}}

\ECSwitch
\ECDisclaimer

\section{Further Related Work}
\label{apx:related work}
\input{Paper/related-work}

\section{\revcolor{Discussion on Practical Relevance}}
\label{apx:micro-foundations}
\input{Paper/apx-practical}


\section{Numerical Experiments}
\label{apx:numerical}
\input{Paper/apx-numerical}

\section{\revcolor{Illustrative Examples for the Interval Assignment Problem} }
\label{apx:IAP illustration}
\input{Paper/apx-IAP-illustration}

\section{\revcolor{Open Problems}}
\label{apx:open-problems}
\input{Paper/apx-open-problem}
\end{document}

%% file: Paper/macros.tex
\TheoremsNumberedThrough 
\EquationsNumberedThrough

\usepackage{tikz}
\usepackage{pgfplots}
\pgfplotsset{compat=newest}
\usetikzlibrary{patterns}
\usepgfplotslibrary{fillbetween}
\usetikzlibrary{intersections}
\usepackage{bbm}

\usepackage[breakable,skins]{tcolorbox}
\definecolor{myLightGray}{gray}{0.9} 

\DeclareRobustCommand{\myboxtwo}[2][gray!15]{
\begin{tcolorbox}[ 
        colback=white,      
        colframe=gray,  
        boxrule=0.2pt,      
        arc=2pt,outer arc=2pt,
        left=12pt,
        right=12pt,
        top=5pt,
        bottom=5pt,
        width=1.07\linewidth,
        enlarge left by=-0.55cm,
        before upper=\renewcommand{\baselinestretch}{1.3}\selectfont,
        after upper=\normalfont
        ]
 #2
 \end{tcolorbox}
}

\usepackage{csquotes}

\renewenvironment*{displayquote}
  {\begingroup
   \setlength{\leftmargini}{0.2cm}%
   \linespread{1.4}\selectfont 
   \csq@getcargs{\csq@bdquote{}{}}}
  {\csq@edquote\endgroup}
\makeatother

\usepackage{xfrac}
\usepackage{xcolor}
\usepackage{bbm}
\usepackage{booktabs}
\usepackage{subfig}
\usepackage[ruled,vlined,linesnumbered]{algorithm2e}
\SetNoFillComment
\usepackage{thm-restate}

\usepackage{mathtools}
\definecolor{cornellred}{rgb}{0.7, 0.11, 0.11}
\definecolor{maroon}{rgb}{0.52, 0, 0}
\definecolor{dgreen}{rgb}{0.0, 0.5, 0.0}
\definecolor{ballblue}{rgb}{0.13, 0.67, 0.8}
\definecolor{royalblue(web)}{rgb}{0.25, 0.41, 0.88}
\definecolor{bleudefrance}{rgb}{0.19, 0.55, 0.91}
\definecolor{royalazure}{rgb}{0.0, 0.22, 0.66}
\usepackage{hyperref}
\hypersetup{
	colorlinks = true,
	linkcolor=royalazure,
	urlcolor=royalazure,
	citecolor=royalazure,
	linkbordercolor = {white}
}
\usepackage{cleveref}
\usepackage{enumitem}
\usepackage{multirow}

\usepackage{pgf,tikz,pgfplots}
\pgfplotsset{compat=1.15}
\usetikzlibrary{arrows}
\usepackage{pgfplots}
\usetikzlibrary{patterns}
\usepgfplotslibrary{fillbetween}
\usetikzlibrary{intersections}

\usepackage{fix-cm}

\usetikzlibrary{arrows, decorations.markings}

\tikzstyle{vecArrow} = [thick, decoration={markings,mark=at position
	1 with {\arrow[semithick]{open triangle 60}}},
	double distance=1.4pt, shorten >= 5.5pt,
	preaction = {decorate},
postaction = {draw,line width=1.4pt, white,shorten >= 4.5pt}]
\tikzstyle{innerWhite} = [semithick, white,line width=1.4pt, shorten >= 4.5pt]

	\usepackage[authoryear,round]{natbib}

\makeatletter
\newcommand{\subalign}[1]{%
  \vcenter{%
    \Let@ \restore@math@cr \default@tag
    \baselineskip\fontdimen10 \scriptfont\tw@
    \advance\baselineskip\fontdimen12 \scriptfont\tw@
    \lineskip\thr@@\fontdimen8 \scriptfont\thr@@
    \lineskiplimit\lineskip
    \ialign{\hfil$\m@th\scriptstyle##$&$\m@th\scriptstyle{}##$\hfil\crcr
      #1\crcr
    }%
  }%
}
\makeatother

	\providecommand{\given}{}
	\DeclarePairedDelimiterX{\set}[1]\{\}{\renewcommand\given{\nonscript\:\delimsize\vert\nonscript\:\mathopen{}}#1}
	\let\Pr\relax
	\DeclarePairedDelimiterXPP{\Pr}[1]{\mathbb{P}}[]{}{\renewcommand\given{\nonscript\:\delimsize\vert\nonscript\:\mathopen{}}#1}
	\DeclarePairedDelimiterXPP{\Ex}[1]{\mathbb{E}}[]{}{\renewcommand\given{\nonscript\:\delimsize\vert\nonscript\:\mathopen{}}#1}

\newcolumntype{P}[1]{>{\centering\arraybackslash}c{#1}}
	
\usepackage{color-edits}
\addauthor{yf}{purple}    

\newcommand{\xhdr}[1]{\vspace{2mm} \noindent{\bf #1}}

\newcommand{\bench}[1]{\textbf{OPT}\left[#1\right]}
\newcommand{\tpre}{t'}
	
	\newcommand{\primed}{^\dagger}

	\newcommand{\Z}{\mathbb{Z}}
	\newcommand{\inventory}{c}
	\newcommand{\mininventory}{{\inventory_{0}}}

	\newcommand{\totaltime}{T}
	\newcommand{\reward}{r}
	
	\newcommand{\choice}{\phi}
	\newcommand{\pen}{\Psi}

	\newcommand{\assortment}{S}
	\newcommand{\assortmentspace}{\mathcal{S}}
	\newcommand{\duration}{d}

	\newcommand{\inventorydual}{\theta}
	\newcommand{\inventorydualijt}{\inventorydual_{i,j,t}}
	
	\newcommand{\probdual}{\lambda}
	\newcommand{\probduali}{\probdual_{t}}

	\newcommand{\algo}{\mathcal{A}}

	\newcommand{\Rev}[2][]{\text{\bf Rev}\ifthenelse{\not\equal{}{#1}}{_{#1}}{}\!\left[{\def\givenn{\middle|}#2}\right]}

	\newcommand{\alloc}{\mathcal{X}}
	\newcommand{\curalloc}{\alloc_{\indexnew, t}}

 \newcommand{\competitiveratio}{\Gamma}

	\newcommand{\noaccents}[1]{#1}
	
	\newcommand{\newagentvar}[3][\noaccents]{%
		\expandafter\newcommand\expandafter{\csname #2\endcsname}{#1{#3}}%
		\expandafter\newcommand\expandafter{\csname #2s\endcsname}{#1{\boldsymbol{#3}}}%
		\expandafter\newcommand\expandafter{\csname #2smi\endcsname}[1][i]{#1{\boldsymbol{#3}}_{-##1}}%
		\expandafter\newcommand\expandafter{\csname #2i\endcsname}[1][i]{#1{#3}\agind[##1]}%
		\expandafter\newcommand\expandafter{\csname #2ith\endcsname}[1][i]{#1{#3}_{(##1)}}%
	}
	\newagentvar{typespace}{{\cal Z}}
	\newagentvar{typesubspace}{S}
	
	\newagentvar{type}{z}
	\newagentvar{othertype}{s}
	\newagentvar{val}{v}
	\newagentvar{hval}{\bar \val}
	\newagentvar{hbudget}{\bar \wealth}
	\newagentvar{budget}{B}
	\newagentvar{lbudget}{\underaccent{\bar}{ \wealth}}
	\newagentvar{lowestval}{0}
	\newagentvar{cumval}{V}
	\newagentvar{cumprice}{P}
	\newagentvar{welcurve}{W}
	\newagentvar{revcurve}{R}
	\newagentvar{outcome}{w}
	\newagentvar{outcomespace}{{\cal W}}

	\newcommand{\exponential}{\pen(x) = \frac{e^{1-x}-e}{1 - e}}

\newcommand*{\rom}[1]{\expandafter\romannumeral #1}
\newcommand{\Rom}[1]{\uppercase\expandafter{\romannumeral #1\relax}}

\newcommand{\contain}{{\mathcal N}}

\newcommand{\containni}{|\mathcal{N}_i|}
\newcommand{\assign}{{\hat I}}

\newcommand{\remaining}{{\mathcal W}}
\newcommand{\remainingcontain}{\contain^\remaining}
\newcommand{\remainingcontainni}{|\contain_{i}^\remaining|}
\newcommand{\pfather}{p}
\newcommand{\pmother}{q}

\newcommand{\replen}{\xi}
\newcommand{\repleni}{\replen_{i, t}}

\newcommand{\naturals}{\mathbb{N}}

\newcommand{\indexset}{\mathcal L}
\newcommand{\indexsett}{\indexset_t}
\newcommand{\indexsetijt}{\indexset_{i,j,t}}
\newcommand{\indexnew}{L}
\newcommand{\addtime}{\tau}
\newcommand{\addtimeij}{\addtime_{i,j}}

\newcommand{\IBB}{Batched Inventory Balancing}
\newcommand{\IB}{Inventory Balancing}
\newcommand{\CUIB}{Dynamic-Capacity Inventory Balancing}
\newcommand{\CNUIB}{Static-Capacity Inventory Balancing}
\newcommand{\UDIB}{Unit-Specific Inventory Balancing}

\newcommand{\batchsizethreshold}{\gamma}

\newcommand{\batch}{{B}}
\newcommand{\batchij}{\batch_{i,j}}

\newcommand{\cnt}{h}
\newcommand{\pseudoinventoryratio}{g}
\newcommand{\pseudoinventoryratioijt}{\pseudoinventoryratio_{i,j,t}}
\newcommand{\inventoryratio}{f}

\newcommand{\inventoryratioijt}{\inventoryratio_{i,j,t}}

\newcommand{\curinventory}{I}

\newcommand{\timeset}{\mathcal{T}}
\newcommand{\timesetij}{\timeset_{i,j}}

\newcommand{\event}{\mathcal E}

\newcommand{\matchinginstance}{G}
\newcommand{\supply}{s}
\newcommand{\supplies}{U}
\newcommand{\demands}{V}

\newcommand{\LocalDominance}{\textsc{Local-Dominance}}
\newcommand{\GlobalDominance}{\textsc{Global-Dominance}}
\newcommand{\PartitionMonotonicity}{\textsc{Partition-Monotonicity}}

%% file: Paper/math.tex
\newcommand{\condition}{\,\mid\,}

\newcommand{\prob}[2][]{\text{Pr}\ifthenelse{\not\equal{}{#1}}{_{#1}}{}\!\left[{\def\givenn{\middle|}#2}\right]}
\newcommand{\expect}[2][]{\mathbb{E}\ifthenelse{\not\equal{}{#1}}{_{#1}}{}\!\left[{\def\givenn{\middle|}#2}\right]}

\newcommand{\tparen}{\big}
\newcommand{\tprob}[2][]{\text{Pr}\ifthenelse{\not\equal{}{#1}}{_{#1}}{}\tparen[{\def\given{\tparen|}#2}\tparen]}
\newcommand{\texpect}[2][]{\mathbb{E}\ifthenelse{\not\equal{}{#1}}{_{#1}}{}\tparen[{\def\given{\tparen|}#2}\tparen]}

\newcommand{\sprob}[2][]{\text{Pr}\ifthenelse{\not\equal{}{#1}}{_{#1}}{}[#2]}
\newcommand{\sexpect}[2][]{\mathbb{E}\ifthenelse{\not\equal{}{#1}}{_{#1}}{}[#2]}

\newcommand{\dd}{{\mathrm d}}

\newcommand{\indicator}[1]{{\mathbbm{1}\left\{ #1 \right\}}}

%% file: Paper/abstract.tex
In classic adversarial online resource allocation problems, such as matching, AdWords, or assortment planning, customers (demand) arrive online, while products (supply) are given offline with a fixed initial inventory. To ensure acceptable revenue guarantees given the uncertainty in future customer arrivals, the decision maker must balance consumption across different products. Motivated by this, the famous policy---\emph{inventory balancing (IB)}---is introduced and studied in the literature and has proved to be optimal or near-optimal competitive in almost all classic settings. However, an important feature that these classic models do not capture are various forms of possible \emph{inventory shocks} on the supply side, which plays an important role in several real-world applications of online assortment and could significantly impact the revenue performance of the IB algorithm.

\revcolor{Motivated by this paradigm, we introduce and study a new variant of online assortment planning with inventory shocks. Our model considers both adversarial exogenous shocks (where supply increases over time in unpredictable fashion) and allocation-coupled endogenous shocks (where an inventory reduction is triggered by the algorithms and is re-adjusted after a usage duration)---combination of which can cause non-monotonic inventory fluctuations. As our main result, we show the robustness of the inventory balancing strategy against inventory shocks and such fluctuations by designing a new family of optimal competitive IB-type algorithms, called \emph{Batched Inventory Balancing (BIB)}.  We develop a novel randomized primal-dual method to bound the competitive ratio of our BIB algorithm against any feasible policy. We show that with the proper choice of a certain parameter, this competitive ratio is asymptotically optimal and converges to $1-1/e$ as the initial inventories converge to infinity---in contrast to the original IB which no longer achieves the optimal competitive ratio in this new model. Moreover, we characterize BIB's competitive ratio in its \emph{general} form (parametric by the penalty function it uses) and show that it matches \emph{exactly} the competitive ratio of IB in the special case without shocks. To this end, we use a refined analysis that reduces the dual construction to a combinatorial problem called the ``interval assignment problem.'' Our solution to this problem is algorithmic and might be of independent interest.}

%% file: Paper/intro.tex
\label{sec:intro}
\revcolor{Online assortment planning is a central problem in operations research, rooted in classical online resource allocation problems such as online matching~\citep{KVV-90} and AdWords~\citep{MSVV-05}. Besides its theoretical appeal, this problem has also broad applications in the revenue management of retail platforms, as well as digital display and search advertising marketplaces. A common assumption in these applications is that each unit of a resource, once allocated and consumed, permanently leaves the inventory. Consequently, the planner begins with a fixed upfront inventory for different resources (e.g., advertisers with fixed daily budgets), which depletes over time due to allocations (e.g., once an impression is assigned to an advertiser, the budget decreases by the amount of the advertiser’s bid, and this money never returns to the advertiser). However, in several more modern online platforms---ranging from cloud computing and short- or long-term homestay marketplaces to online volunteer matching and freelance marketplaces--- this assumption could be violated as inventory levels fluctuate over time due to various forms of \emph{endogenous} and \emph{exogenous} shocks. }

\revcolor{An endogenous shock is typically triggered by an allocation decision, for example, when a user selects a reusable rental resource such as a cloud server on AWS or a listing on Airbnb~\citep{GGISUW-19, GIU-20, FNS-24,DFNSU-24}. Similarly, in volunteer matching platforms, when a volunteer is nudged to perform a newly arriving task, they become unavailable for a certain period of time---or the platform stops nudging them for some time, to avoid overwhelming the volunteers~\citep{MR-22}. In all such cases, a unit of an offline resource temporarily leaves the inventory after an endogenous shock and returns after a period of usage. In contrast, an exogenous shock refers to an unpredictable and non-stationary replenishment of the inventory.\footnote{\revcolor{One may also consider negative exogenous shocks. Although these are less common than positive shocks in motivating applications of our model, studying them remains an interesting direction; see \Cref{sec:conclusion} and \Cref{apx:open-problems} for further discussion.\label{footnote:FN1}}} For instance, in omni-channel retail, unexpected product returns may occur when customers who purchased through other channels find the product does not meet their expectations, or when there is misalignment between inventory management and sale~\citep{dBD-03}.\footnote{\revcolor{While returning a defective or undesired product occurs after the original sale, the presence of multiple return channels and the heterogeneity and non-stationarity of returns motivate an adversarial modeling; see \cite{BCDL-10} for discussion.\label{footnote:FN2}}} Another example arises in volunteer platforms such as \href{https://www.catchafire.org/discover/}{Catchafire} and \href{https://foodrescue.us/}{Food Rescue U.S.}, where volunteers initially commit to a limited capacity but may later decide to spend additional working hours (thanks to engagement dynamics and incentive policies), allowing them to serve more than their initial commitment~\citep{LMRS-24}. In all these cases, one or several units of an offline resource are added to the inventory after an unpredictable exogenous shock and become available for future allocations.\footnote{\revcolor{See \Cref{apx:micro-foundations} for a detailed discussion of various micro-foundations and applications of endogenous and exogenous shocks.}}}

Motivated by the above \revcolor{paradigm}, we introduce and study the problem of \emph{online assortment with inventory shocks}. We consider a sequential allocation problem in which a sequence of consumers with arbitrary choice functions arrives over time. At each time step, the planner selects a subset of available products to display to the arriving consumer. The consumer then probabilistically chooses one of the displayed products according to her choice model, consumes one unit of that product, and pays a product-dependent price. The planner faces both endogenous and exogenous inventory shocks. \revcolor{For endogenous shocks, once a unit of a product is purchased, it returns to the inventory after a fixed duration of usage, which may differ across products but remains constant over time.} For exogenous shocks, we assume that at the beginning of each time step, a batch of products may be added to the inventory. \revcolor{Consumer choices and exogenous shocks are assumed to be arbitrary—chosen by an \revcolor{(adaptive)} adversary—and are revealed to the platform only at the start of each time step, before the assortment is selected.} In the absence of exogenous shocks, our model reduces to the online assortment with \emph{reusable resources}, previously studied in the revenue management literature~\citep{RST-17,GGISUW-19,GIU-20,FNS-24}.

The goal of the assortment planner is to sequentially select inventory-feasible assortments while facing both endogenous and exogenous inventory shocks, in order to maximize the expected collected price (which we refer to as ``revenue''). In principle, the one-shot assortment planning problem can be a  computationally hard set-function optimization. To address this challenge, following the related literature~\citep{GNR-14,GGISUW-19,FNS-24,GIU-20}, we assume an oracle access to an offline assortment solver given the consumer choice model. \revcolor{We also assume that consumers' choice models satisfy the weak substitutability assumption---meaning that removing a product from the assortment can only weakly increase the probability of selecting another product---and the monotonicity assumption---meaning that adding a product to the assortment can only decrease the probability of selecting the outside option.}


We aim to design online competitive algorithms for our problem and compare their performance to that of the optimal offline benchmark, which has complete knowledge of consumer choices and inventory shocks but not the exact choice realizations. This benchmark, commonly referred to as the \emph{clairvoyant optimum} in the literature, provides an upper bound on the expected revenue achievable by any policy. We say that an algorithm is \emph{$\Gamma$-competitive} for $\Gamma \in [0,1]$ if, for every instance, the expected revenue of the algorithm is at least $\Gamma$ times that of the clairvoyant optimum. We study the \revcolor{dependence} of the competitive ratio on the minimum initial inventory across all products, denoted by $\mininventory$, and focus primarily on the regime in which $\mininventory$ is large. Importantly, we make no assumptions on the exogenous shocks, \revcolor{other} than that they are non-negative.

\xhdr{Summary of our contributions.} Our work is inspired by the simple and elegant classic \emph{BALANCE} algorithm—first introduced in \citet{KP-00} for the online bipartite $b$-matching problem, later generalized to the AdWords problem in \citet{MSVV-05}, and subsequently to the non-reusable online assortment problem in \citet{GNR-14}. This algorithm, also referred to as \emph{inventory balancing (IB)} in the assortment planning context, achieves the optimal $\left(1 - 1/e\right)$ competitive ratio in the fractional versions of these problems, or equivalently, in their integral versions under the “large inventory” regime (e.g., when $\mininventory$ is large in online assortment). Moreover, \citet{GNR-14} consider a parametric version of IB (as a family of algorithms), where each algorithm is parameterized by a concave and increasing penalty function $\pen:[0,1]\rightarrow [0,1]$ used to balance inventories. They provide an asymptotically tight characterization of the competitive ratio $\Gamma$ of an IB algorithm, given the penalty function $\pen$ and the minimum initial inventory $\mininventory$, in the absence of endogenous or exogenous inventory shocks.

The above results are of paramount importance in the related application domains, as BALANCE/IB is simple, easy to implement, flexible to adapt to specific problem instances through the appropriate choice of penalty function, and transparent in its operations. Motivated by these properties, our goal is to address the following research question concerning the robustness of inventory balancing:
\vspace{-1mm}
\myboxtwo{
\begin{displayquote}
\emph{\revcolor{Is there a natural adaptation of the IB online algorithm that achieves the same asymptotic and non-asymptotic optimal competitive ratios established in \citet{GNR-14}, as if the online algorithm were not subject to endogenous and exogenous inventory shocks?}}
\end{displayquote}}
\vspace{-2mm}
Our main result provides a compelling answer to the above question. We propose a new parametric family of online assortment algorithms, referred to as \emph{batched inventory balancing (BIB)}, which extends the original “balance’’ idea to online assortment problems with both endogenous and exogenous inventory shocks. Using a novel randomized primal–dual analysis, we prove that this family achieves the asymptotically optimal competitive ratio of $1 - 1/e$ in the large-inventory regime. We further extend this result by characterizing an almost tight competitive ratio $\Gamma$ for the general case, expressed as a function of the penalty function $\pen$ (a concave function analogous to that in the vanilla IB algorithm) and the minimum initial inventory $\mininventory$ in the non-asymptotic regime (see \Cref{thm:concave competitive ratio} for details). Notably, our bound exactly matches the one established in \citet{GNR-14} for online assortment without inventory shocks.

\revcolor{The simple idea behind our algorithm is to treat different units of the same product separately. More precisely, for each product, we \emph{assign} its units into what we call \emph{batches} based on the time that a unit is added into the inventory due to an exogenous shock (or as part of the initial inventory). Each batch can be in one of two states: if it has accumulated a sufficiently large number of units, its available units can be used for the assortment---such a batch is called \emph{``ready''}; otherwise, if it has not yet reached this size threshold, its available units cannot be used for the assortment until it grows larger---such a batch is called \emph{``charging''}.}

\revcolor{Our BIB algorithm operates by tracking (i) the assignment of units to ``ready'' batches and the potential ``charging'' batch for each product, and (ii) the inventory levels of the ``ready'' batches across all products.} At each time step, when a new consumer arrives, the algorithm penalizes the price of each product based on the highest normalized inventory level among its ``ready’'' batches. It then solves a one-shot assortment problem by making a customized query call to the oracle\revcolor{, where prices are replaced with penalized prices.} When the consumer chooses a specific product in that assortment, the algorithm then allocates one unit of that product from the ``ready'' batch with the highest normalized inventory level.

\xhdr{New analysis framework: randomized primal–dual \& interval assignment problem.} We rely on a primal–dual analysis to establish the bounds on the competitive ratio of BIB. Rather than comparing directly to the clairvoyant benchmark, we follow prior work and consider an offline linear programming (LP) relaxation of this policy. This LP benchmark has access to the same information as the clairvoyant policy about the input instance, but only requires that assortments respect inventory feasibility in expectation for each product unit. \revcolor{Unlike the classic LP benchmarks in the literature, it is important to note that our offline LP benchmark (named \emph{``batch-specified LP''} and formulated in \ref{eq:opt}, where $\gamma$ is the batch size threshold of the BIB algorithm) also depends on the execution of the BIB algorithm itself. In fact, similar to the BIB algorithm, this LP benchmark treats different units from different batches separately.}

Our primal–dual method is randomized and, at least superficially, resembles the dual-fitting analysis of \citet{DJK-13} for the RANKING algorithm, originally introduced in the seminal work of \citet{KVV-90} for the online bipartite matching problem. What makes the analysis of BIB particularly challenging—and distinguishes it from prior work—is twofold:
(i) the times at which different units of the same product are added to the inventory due to \revcolor{exogenous shocks} may vary, and since these units are treated differently by the primal BIB algorithm, they require distinct dual assignments; and
(ii) the inventory level of an arbitrary batch is not necessarily \emph{monotone decreasing} because of both endogenous and exogenous inventory shocks, in contrast to the classic online bipartite allocation problem without inventory shocks. \revcolor{This loss of monotonicity, a property central to previous analyses, introduces significant new challenges in our setting.}


To address challenge (i), \revcolor{we develop the novel batch-specified LP benchmark~\ref{eq:opt}, which, as alluded to earlier, is defined based on the execution path of the BIB algorithm.} Furthermore, our primal–dual analysis constructs the duals at the batch level rather than at the product level (\Cref{sec:IBB proof}). To address challenge (ii), we employ a randomized dual construction that incorporates a \emph{transformed version of the inventory levels} of “ready’’ batches through a particular combinatorial mapping. This mapping must satisfy certain structural properties to ensure that the resulting dual assignment is (approximately) feasible in expectation. We then restate our dual feasibility conditions in a purely combinatorial form and introduce a stand-alone combinatorial problem, which we call the \emph{Interval Assignment Problem (IAP)}; see \Cref{lem:IAP} and \Cref{sec:IAP} for more details. Our final dual construction is obtained by developing an algorithmic solution to IAP (\Cref{alg:interval assignment}).


The IAP and our analysis framework could be valuable tools for other related problems. In the literature on online resource allocation, various penalty functions beyond the exponential penalty function have been studied and shown to be optimal or maintain the state of the art in various models \citep[e.g.,][]{HTWZ-17,MS-19,JW-21,JM-22,EFN-23}. When incorporating these models with inventory shocks, our framework for general penalty functions, including the IAP, might be applicable. We leave any further developments to future research.

Finally, one important implication of our result is the \emph{robustness of the performance guarantee} of the (batched) inventory balancing algorithm, even in complex scenarios where inventories change in non-monotone fashion due to both endogenous and exogenous inventory shocks. This is particularly relevant in applications, as simplicity and transparency of this algorithm, as well as the fact that it does not need to keep track of the full state of the restocking, makes it more practically appealing.

For a more in-depth discussion of related work, see \Cref{apx:sec:further-related} in the electronic companion.


%% file: Paper/prelim.tex
\label{sec:prelim}
\xhdr{Problem formulation.}
We consider an online assortment planning problem subject to inventory constraints and inventory shocks. The platform manages inventories for $n$ different products, indexed by $[n] = \{1, 2, \dots, n\}$. Each product $i$ has an initial inventory of $\inventory_i \in \Z_{+}$ units, \revcolor{a deterministic \emph{usage duration} $\duration_i$ (also referred to as the return duration),} and a per-unit price $\reward_i$. We study a discrete, finite-horizon setting with $\totaltime$ time periods, during which product inventories may experience both endogenous and exogenous shocks.\footnote{\label{footnote:continuous time extension}\revcolor{Our results also extend to \emph{continuous} (i.e., real-valued) time. At each real-valued time $t \in [0,T]$, exogenous shocks $\{\replen_{i,t}\}_{i \in [n]}$ and a consumer may arrive. If the consumer selects product $i$, one unit is sold and becomes available again at real-valued time $t + \duration_i$.}}

For endogenous shocks, if a unit of product $i$ is sold at time $t \in [\totaltime]$, it remains under usage for $\duration_i$ periods and returns to the inventory at time $t + \duration_i$. For exogenous shocks, we assume that $\repleni \in \Z_{\geq 0}$ units of each product $i$ are added to the inventory at the beginning of period $t$. The sequence of exogenous shocks $\{\repleni\}_{i \in [n], t \in [\totaltime]}$ is unknown to the platform in advance and is chosen by an adaptive adversary. We distinguish between the initial inventories $\{\inventory_i\}_{i\in[n]}$ and the first-period shocks $\{\replen_{i,1}\}_{i\in[n]}$.\footnote{We typically consider the regime in which the initial inventories $\{\inventory_i\}$ are large, though we also analyze the small-inventory case. Nonetheless, we impose \emph{no restrictions} on the exogenous shocks $\{\repleni\}_{i \in [n], t \in [\totaltime]}$ other than non-negativity. The competitive ratio of our algorithm remains asymptotically optimal as $\min_i \inventory_i$ tends to infinity.}


Consumers who are interested in buying these products arrive sequentially in discrete times $t\in[\totaltime]$.
Each consumer $t$ has a choice model $\choice_t:
\assortmentspace
\times [n] \rightarrow [0,1]$,
where 
$\assortmentspace\subseteq 2^{[n]}$ is the collection of all feasible assortments that can be offered,
and
$\choice_t(\assortment, i)$ is 
the probability that consumer $t$ chooses product $i$ 
when assortment set $S\in\assortmentspace$ is offered. 
Upon the arrival of consumer $t$,  
her choice model $\choice_t$ 
and
the adversarial exogenous shocks $\{\repleni\}_{i\in[n]}$
are revealed to the platform.  
Given $\choice_t$ and the history up to time $t$, 
the platform offers an assortment
$\assortment_t \in \assortmentspace$ of products from its available inventory.
Based on the choice probabilities $\choice_t(\assortment_t,i)$ for $i\in [n]$, consumer $t$ randomly selects a product $i_t\in \assortment_t$ or selects nothing. If $i_t$ is selected, then she pays the price $\reward_{i_t}$
to the platform. This unit sold returns to inventory at the beginning of time $t+d_i$.

We assume that consumers' choice models are chosen by an adversary (oblivious or adaptive). We further impose two assumptions on our choice models and feasible assortments, which have been commonly used in the past literature  \citep[][]{GNR-14, RST-17}:

\begin{assumption}[\emph{Weak substitutability \& monotonicity}]
	\label{asp:substitutability}
	For all time periods $t\in[\totaltime]$,
	choice models $\choice_t$ and products 
	$i \in [n]$,
	$\choice_t (\emptyset, i) = 0$. Moreover, for all 
	$\assortment\in \assortmentspace$ and $j\in [n]/ \{i\}$, $\choice_t (\assortment, i)\geq \choice_t (\assortment\cup\{j\}, i).$
\end{assumption}
\begin{assumption}[\emph{Downward-closed feasibility}] 
\label{asp:down-closed}
If $S\in\assortmentspace$ and $S'\subseteq S$ then $S'\in\assortmentspace$, i.e., a feasible assortment will remain feasible after removing any subset of its offered products. 
\end{assumption}

\xhdr{Competitive ratio analysis.} \revcolor{We evaluate the performance of the online algorithms compared to the \emph{clairvoyant optimum offline benchmark,} which is the optimum offline solution maximizing expected revenue, knowing the sequence of consumer choices and exogenous shocks $\{\choice_t, \repleni\}_{i\in[n],t\in[\totaltime]}$ but not knowing the exact realizations of the stochastic choices. More specifically, we study the worst case \emph{competitive ratio} of our algorithm against this benchmark, defined formally as follows:}
\begin{definition}[Competitive ratio]
	\label{def:competitive ratio}
	Fix $n$ products with 
	parameters $\{\reward_i, \duration_i,
	\inventory_i\}_{i\in[n]}$.
	An online algorithm $\algo$ is \emph{$\alpha$-competitive} \revcolor{against the clairvoyant optimum offline benchmark if}
	\begin{align*}
		\inf\limits_{T \geq 1}
		\inf\limits_{
		\{\choice_t, \repleni\}_{i\in[n],t\in[\totaltime]}
		}
		\frac{\Rev[\algo]{\{\choice_t, \repleni\}_{i\in[n],t\in[\totaltime]}
		}}{\bench{
		\{\choice_t, \repleni\}_{i\in[n],t\in[\totaltime]}
		}} 
		\geq \alpha~,
	\end{align*}
	\revcolor{where $\Rev[\algo]{\cdot}$ and $\bench{\cdot}$ are the expected revenues of algorithm $\algo$ and the clairvoyant optimum offline benchmark, respectively (with input being the sequence of consumer choices and exogenous shocks).}
  
\end{definition}
\revcolor{In simple words, the competitive ratio above is defined as the worst-case ratio between the expected total revenue of the online algorithm and the clairvoyant optimum offline benchmark, where the worst case is over all sequences of consumer choice models and exogenous inventory shocks over time.}\footnote{\label{footnote:OPT definition}\revcolor{Notably, due to exogenous inventory shocks, we introduce and analyze a batch-specified linear-programming relaxation of the clairvoyant optimum offline (formally defined in \Cref{sec:benchmark}), as a proxy to obtain competitive ratios against this benchmark.}}

\xhdr{Oracle access to offline assortment solver.} Similar to prior work, we assume having access to an oracle algorithm that solves the one-shot offline assortment problem. For simplicity, we assume the solver is exact throughout the paper, but all of our results still hold with a multiplicative loss
of $\beta$ in the competitive ratios if the solver is a $\beta$-approximation algorithm for some $0<\beta<1$. 
\begin{assumption}[\emph{Offline assortment oracle}]
	\label{asp:oracle}
	For all $t\in[\totaltime]$ and non-negative reward $\mathbf{\hat{R}}\in\mathbb{R}_+^n$, we have oracle access to an algorithm that finds an assortment $\hat{\assortment}\in\assortmentspace$ such that:
$\hat{\assortment}\in\argmax\nolimits_{\assortment\in\assortmentspace}~{\sum\nolimits_{i=1}^n\hat{R}_i\choice_t(\assortment,i)}$ and has the smallest cardinality among all such reward-maximizing assortments (for tie-breaking). 
\end{assumption}

%% file: Paper/main-result.tex
\label{sec:IBB}
In this section, we propose our new online assortment algorithm, termed \emph{batched-inventory-balancing (BIB)}. \revcolor{We also summarize our main result, which shows that with appropriately chosen parameters, the BIB algorithm is not only asymptotically optimal in terms of its competitive ratio but also achieves a competitive ratio that \emph{exactly} matches that of the vanilla Inventory Balancing (IB) algorithm in online assortment optimization, as established in the setting without endogenous or exogenous shocks~\citep{GNR-14}. We prove this result in the next section (\Cref{sec:CR}) by developing a novel and refined primal–dual analysis for our BIB algorithm, which may be of independent interest.}



\subsection{An Overview of the BIB Algorithm}
The BIB algorithm begins with a simple idea. Instead of monitoring the total inventory level---that is, the total number of available units for each product~$i$---BIB treats the units of product~$i$ individually and dynamically assigns them to \emph{batches} as new units are added to the inventory, either through exogenous shocks or as part of the initial inventory. Let $\batchij$ denote the $j^{\textrm{th}}$ batch for product~$i$. The construction of batches $\{\batchij\}_{j=1,2,\dots}$ for a specific product $i \in [n]$ \revcolor{proceeds} as follows. At time period~0, all initial units of the product are grouped into the first batch, $\batch_{i,1}$, which is marked as ``\emph{\underline{ready}}''. Simultaneously, an empty batch, $\batch_{i,2}$, is created and marked as ``\emph{\underline{charging}}''. For each subsequent time period $t \in [T]$, any new unit of product~$i$ added through an exogenous inventory shock is permanently assigned to the current ``charging'' batch. 

The BIB algorithm uses a predetermined \emph{batch-size threshold} $\batchsizethreshold \in \mathbb{N}$. When the number of units in the current ``charging'' batch, denoted by $|\batchij|$, reaches this threshold, the batch is switched to ``ready,'' and a new empty batch marked as ``charging'' is created. Importantly, when a sold unit of a product returns to the inventory (an endogenous shock), it is added back to the same batch to which it was originally assigned---either when it first entered the inventory through an exogenous shock or as part of the initial inventory.

The BIB algorithm is also parameterized by a \emph{penalty function} $\pen:[0,1]\rightarrow[0,1]$ satisfying $\pen(0)=0$ and $\pen(1)=1$. During each time period~$t$, while performing the batch construction process described above for each product upon the arrival of exogenous shocks, \revcolor{the BIB algorithm maintains a record of the \emph{normalized inventory levels} $\{\inventoryratioijt\}$, where $\inventoryratioijt$ denotes the fraction of units of product~$i$ currently available at time $t$ among all units assigned to the ``ready'' batch~$\batchij$ before and at time~$t$.}\footnote{\revcolor{Let $\event(i, j, \tpre)$ denote the event that an available unit from product~$i$’s “ready’’ batch~$\batchij$ is allocated (or sold) at time~$\tpre$. Then the normalized inventory level is defined as $\inventoryratioijt \triangleq 1 - \sum_{\tpre \in [t - \duration_i + 1 : t - 1]} \indicator{\event(i, j, \tpre)} / |\batchij|$.}} The algorithm then uses these normalized inventory levels to determine the final assortment. Specifically, it selects the assortment $\assortment_t$ by invoking the offline assortment oracle with a reduced price $\hat{R}_i \gets \reward_i \pen(\max_j \inventoryratioijt)$ for each product~$i$, where the maximization is taken over all “ready’’ batches of that product. If product~$i^*$ is chosen by consumer~$t$, then an available unit from the “ready’’ batch~$\batch_{i^*, j^*}$ with $j^* = \argmax\nolimits_j \inventoryratio_{i^*, j, t}$ is allocated.\footnote{There always exists an available unit in batch~$\batch_{i^*, j^*}$. To see this, suppose otherwise. Then $\inventoryratio_{i^*, j^*, t} = 0$, and the reduced price $\hat{R}_{i^*}=0$ since $\pen(0)=0$. Hence, the assortment $\assortment_t \setminus \{i^*\}$, which remains feasible by Assumption~\ref{asp:down-closed}, would also maximize the reduced total revenue, contradicting the tie-breaking rule in Assumption~\ref{asp:oracle}.} See \Cref{alg:IBB} for a formal description of the BIB algorithm.

\input{Paper/alg-BIB}

\subsection{Main Result: Competitive Ratio of BIB}
Our main theorem establishes competitive ratio bounds for the general BIB algorithm 

\begin{restatable}{theorem}{IBratio}
	\label{thm:concave competitive ratio}
	For any monotone non-decreasing, concave and differentiable
	penalty function $\pen$, and
	any batch-size scalar \revcolor{$\batchsizethreshold\in[1, \mininventory]$},
	the competitive ratio of the 
	\IBB\ algorithm 
	is at least 
	$\competitiveratio(\pen, \batchsizethreshold) = 
	\min\{\competitiveratio_1(\pen,\batchsizethreshold), 
	\competitiveratio_2(\pen,\batchsizethreshold)\}$, where 
	\begin{align*}
    \competitiveratio_1(\pen,\batchsizethreshold) &= \min_{x\in[0,1-\frac{1}{\mininventory}]}
    \left\{
    \frac{1-x}{\frac{1}{\mininventory} + (1 + \frac{\batchsizethreshold}{\mininventory})(1 - \pen(x)) + \int_{x+\frac{1}{\mininventory}}^1\pen(y)dy}
    \right\}~,
    \\
    \competitiveratio_2(\pen,\batchsizethreshold) &=
    \min_{x\in[0,1-\frac{1}{\batchsizethreshold}]}
    \left\{
    \frac{1-x}{\frac{1}{\batchsizethreshold} + 1 - \pen(x) + \int_{x+\frac{1}{\batchsizethreshold}}^1\pen(y)dy}
    \right\}~,
\end{align*}
and $\mininventory = \min_{i\in[n]}\inventory_i$ is the smallest initial inventory.
\end{restatable}
We provide a comprehensive analysis and prove \Cref{thm:concave competitive ratio} in \Cref{sec:CR}. Before that, we next elaborate on some theoretical and practical corollaries and implications of our main result. 

\subsection{Corollaries and Implications of Theorem~\ref{thm:concave competitive ratio}}
\label{sec:implication}
Based on \Cref{thm:concave competitive ratio}, it can be verified that the competitive ratio of the BIB algorithm is asymptotically optimized by selecting the exponential penalty function $\exponential$ and $\batchsizethreshold = \Theta(\sqrt{\mininventory})$, which approaches $\left(1-\sfrac{1} {e}\right)\approx 63\%$ as $\mininventory$ goes to infinity. This bound, which is now obtained in the presence of endogenous and exogenous shocks, is the asymptotic optimal bound even for the special cases of the online assortment~\citep{GNR-14} and the online vertex-weighted b-matching~\citep{KVV-90,BM-08, AGKM-11} without inventory shocks.


In addition to establishing the asymptotic optimality of the BIB algorithm under the exponential penalty function, the competitive ratio bound for a general concave penalty function $\pen$ in \Cref{thm:concave competitive ratio} has three key implications. First, when compared with the competitive ratio bound of the vanilla IB algorithm with an arbitrary concave penalty function $\pen$---as proved in \citet{GNR-14} for the special case of online assortment optimization \emph{without} inventory shocks---we show that the BIB algorithm achieves the \emph{exact same} asymptotic competitive ratio as IB for \emph{any} choice of concave penalty function $\pen$, even in the presence of both exogenous and endogenous shocks.\footnote{In \Cref{apx:numerical}, we examine several natural generalizations of the vanilla Inventory Balancing (IB) algorithm that are considerably simpler than BIB. We construct specific instances demonstrating that these generalizations fail to achieve the asymptotic optimal competitive ratio of $1 - 1/e$, even in the special case with only exogenous inventory shocks. In particular, we establish competitive ratio upper bounds of $0.552$, $0.53$, and $0.5$ for these simple IB variants in \Cref{apx:numerical stylized instance}.} This result highlights the robustness and power of inventory-balancing-style algorithms in online resource allocation: when properly adapted, they preserve their performance even under dynamic inventory environments with endogenous and exogenous shocks. Second, the precise competitive ratio bound in \Cref{thm:concave competitive ratio} also characterizes the dependence on the initial (finite) inventory levels, offering a sharper perspective beyond the asymptotic regime potentially useful for practitioners. \revcolor{Lastly, \Cref{thm:concave competitive ratio} provides a best-of-both-worlds guarantee: it ensures robust performance under arbitrary concave penalty functions in general settings, while allowing practitioners to leverage application-specific structures by tailoring the penalty function for potentially stronger guarantees~\citep[e.g.,][]{FN-25}.}

%% file: Paper/alg-BIB.tex
\begin{algorithm}[hbt]
 	\caption{\textsc{\IBB~(BIB)}}
 	\label{alg:IBB}
 	\KwIn{penalty function $\pen$, 
 	batch-size threshold $\batchsizethreshold$.}
 	
        \For{each product $i\in[n]$}{
            initialize batch $\batch_{i, 1}\gets [\inventory_i]$ as the ``ready'' batch containing all $\inventory_i$ initial units.

            initialize batch $\batch_{i, 2} \gets \emptyset$ as the ``charging'' batch.
            
            initialize ``ready'' batch counter $\cnt_i\gets 1$.
        }
 	
 	\For{each time period $t\in[T]$}{ 
 	
 	\For{each product $i\in[n]$}{

        \For{each new unit of product $i$ added from exogenous shock}{

 	      assign this new unit into batch $\batch_{i,\cnt_i + 1}$ 
 	
 	\If{$|\batch_{i,\cnt_i + 1}| = \batchsizethreshold$}{
 	update $\cnt_i \gets \cnt_i + 1$
  
 	  initialize $\batch_{i, \cnt_i + 1}\gets \emptyset$ as the new ``charging'' batch.	
 	}
 	}
        }
 	
 	{\color{blue}\tcc{consumer $t$ with choice model $\choice_t$ arrives.}}
  {\color{blue}\tcc{some sold units of products are returned to inventory due to endogenous shocks (i.e., restocking).}}
 	
        update normalized inventory levels $\{\inventoryratioijt\}$ of ``ready'' batches.
        
        solve 
         $\assortment_t \gets \argmax\nolimits_{\assortment\in \assortmentspace}
        \displaystyle\sum\nolimits_{i\in \assortment} 
        \reward_i\pen(\max\nolimits_{j\in[\cnt_i]} \inventoryratioijt)
        \cdot \choice_{t}(\assortment, i)
        $ with offline assortment oracle.

        offer assortment $\assortment_t$ to consumer $t$.

        {\color{blue}\tcc{suppose consumer $t$ chooses product $i^*$.}}

        allocate an available unit of product $i^*$ from batch $\batch_{i^*,j^*}$
        with $j^* = \argmax\nolimits_{j\in[\cnt_i]}\inventoryratio_{i^*,j,t}$.
    }
\end{algorithm}

%% file: Paper/cr-analysis.tex
\label{sec:CR}
As mentioned earlier in the introduction, our goal is to use the primal–dual analysis to prove \Cref{thm:concave competitive ratio}. In contrast to prior work that employs similar approaches, our analysis faces two main technical challenges: (i) different units of the same product require distinct dual assignments, since the times at which they are added to the inventory due to exogenous shocks may differ; and
(ii) the inventory level of an arbitrary batch is no longer monotonically decreasing because of both endogenous and exogenous shocks---a property that was crucial for the success of prior primal–dual analyses in both the online matching problem and its generalizations to online assortment.

\xhdr{Overview of the analysis.} To address the challenge (i), we come up with a batch-specified LP revenue benchmark that relies on the execution path of the BIB algorithm (\Cref{sec:benchmark}), and our primal-dual analysis designs a dual construction at the batch level rather than at the product level (\Cref{sec:IBB proof}). To address the challenge (ii), roughly speaking, the dual construction uses a transformed version of inventory levels of ``ready'' batches using a particular mapping. The mapping is combinatorial and constructing it is rather involved; it uses a chain of techniques to reduce the problem of finding an appropriate inventory mapping for our purpose to solving a sub-problem called  \emph{Interval Assignment Problem (IAP)} (\Cref{lem:IAP} and \Cref{sec:IAP}). Given the solution to IAP, our final dual construction uses transformed inventory levels and incorporates randomness in a specific fashion to provide a certificate for the approximate optimality of the primal algorithm.

\subsection{Batch-Specified Assortment and LP Revenue Benchmark}
\input{Paper/unit-specified-view}

\subsection{Competitive Ratio Analysis of \IBB}

\input{Paper/primal-dual-analysis}

\subsection{Interval Assignment Problem}
\input{Paper/interval-assignment}

%% file: Paper/unit-specified-view.tex
\label{sec:benchmark}
%
\revcolor{We first define the revenue benchmark used in analyzing the competitive ratio of the BIB algorithm (\Cref{alg:IBB}). Recall that the BIB algorithm views different units of the same product as separate entities rather than identical ones. This distinction is important, as different units may become available at different times due to replenishments. Therefore, to make the comparison meaningful, the benchmark should also distinguish between units of a fixed product. With this motivation, we extend the standard concept of an assortment (which is defined over products) to the notion of \emph{batch-specified assortments}. Building on this, we then introduce the \emph{batch-specified LP benchmark}, where both concepts are defined with respect to the product batches generated by the BIB algorithm.}

\xhdr{Batch-specified assortments.} Though all units of a fixed product $i$ have the same price $r_i$ and are identical from the consumers' perspective, it is useful to consider them separately.\footnote{In fact, ignoring endogenous effects of return, one can think of the inventory at each time being formed by the initial inventory and added inventory due to exogenous shocks at different times. With this view, considering units separately (as will be clear later) helps with keeping track of the time they are added for the first time to the inventory.} More precisely, our analysis partitions all units of each product $i$ based on their assigned batches in the BIB algorithm, and units in each \revcolor{batch} are considered separately.

We introduce an important auxiliary concept, which we term as \emph{batch-specified assortment}, defined as follows. A batch-specified assortment $\indexnew = \{(i,j_i): i\in S\}$ encodes a subset of products  $S$ (same as a typical assortment) together with a specific batch $\batch_{i,j_i}$ for each product $i\in\assortment$. With slight abuse of notation, we define $\choice_t(\indexnew,i)\triangleq \choice_t(\assortment, i)$ for every batch-specified assortment $\indexnew =  \{(i,j_i): i\in S\}$. By offering batch-specified assortment $\indexnew$ to consumer $t$, a unit from batch $\batch_{i,j_i}$ of each product $i\in\assortment$ is selected with probability $\choice_t(\indexnew,i)$. 

Recall the (dynamic) construction of batches in the BIB algorithm, described formally in \Cref{alg:IBB}. When the algorithm terminates at the end of the horizon $T$, there are $\cnt_i$ ``ready'' batches for each product $i\in[n]$. We denote by $\addtimeij$ as the time period when batch $\batchij$ is switched from ``charging'' to ``ready''. As a sanity check, $\addtime_{i,1} = 1$ and $\addtime_{i,j}$ is weakly increasing in $j$. We say a batch-specified assortment $\indexnew =  \{(i,j_i): i\in S\}$ is \emph{feasible} in time period $t$ if its encoded assortment $\assortment$ is feasible (i.e., $\assortment\in\assortmentspace$) and $\addtime_{i,j_i} \leq t$ for every product $i\in\assortment$. We use $\indexset_t$ to denote the set of all feasible batch-specified assortments in time period $t$. Moreover, we use $\indexsetijt$ to denote all feasible batch-specified assortments in time period $t$ that contain batch $\batchij$ of product $i$. As a sanity check, the BIB algorithm is essentially offering batch-specified assortment 
\revcolor{$\indexnew_t = \{(i,j_i): i\in\assortment_t, j_i=\argmax_{j}\inventoryratioijt\}$}
for each consumer $t\in[T]$ 
and each batch-specified assortment $\indexnew_t$ is feasible in time period $t$, i.e., $\indexnew_t \in\indexsett$.



\xhdr{Batch-specified LP benchmark.} 
Similar to the batch-specified assortment introduced above, the revenue benchmark for our competitive ratio analysis also relies on the execution of the BIB algorithm, and in particular, the batch formulation. Consider the following linear program that knows the exact sequence of $\{\choice_t, \repleni\}_{i\in[n],t\in[\totaltime]}$, but not the realization of consumers' choices:
\begin{align}
\label{eq:opt}
\tag{$\mathcal{P}_{\textsc{OPT}}(\batchsizethreshold)$}
	\begin{array}{lll}
	{\max
	\limits_{\mathbf \alloc \geq \mathbf 0}}~~&
	\displaystyle\sum\nolimits_{t=1}^\totaltime 
	\displaystyle\sum\nolimits_{\indexnew \in \indexset(t)}
	\displaystyle\sum\nolimits_{i=1}^n
	\reward_i\choice_t(\indexnew, i) \curalloc
	&\text{s.t.}
        \\[1em]
	  &
	\displaystyle\sum\nolimits_{t =t\primed}^{t\primed+\duration_i-1}
	\displaystyle\sum\nolimits_{\indexnew \in \indexset_{i,1,t}}
	\choice_{t}(\indexnew, i)
	\curalloc \leq \inventory_i + \batchsizethreshold \quad\quad
	& i \in [n],\ 
	t\primed \in [\addtime_{i,1}: \totaltime]
	\\[1em]
	  &
	\displaystyle\sum\nolimits_{t =t\primed}^{t\primed+\duration_i-1}
	\displaystyle\sum\nolimits_{\indexnew \in \indexsetijt}
	\choice_{t}(\indexnew, i)
	\curalloc \leq \batchsizethreshold
	& i \in [n],\ 
	j\in[2:\cnt_i],\
	t\primed \in [\addtimeij: \totaltime]
	\\[1em]
	&
	\displaystyle\sum\nolimits_{\indexnew \in \indexsett}
	\curalloc 
	\leq 1
	&t \in [\totaltime]~.
\end{array}
\end{align}
It can be observed that program~\ref{eq:opt} 
\revcolor{depends on the definition of BIB algorithm.}
Specifically, $\batchsizethreshold$ is the batch-size threshold of the BIB algorithm, $\{\cnt_i\}_{i\in[n]}$ is the number of ``ready'' batches in the algorithm, and $\{\addtimeij,\indexsett,\indexsetijt\}$ are all induced by the BIB algorithm. To highlight this dependence on the BIB algorithm, we parameterized program~\ref{eq:opt} by $\batchsizethreshold$.

In program~\ref{eq:opt}, each variable $\curalloc$ can be interpreted as the probability of offering a batch-specified assortment $\indexnew$ for consumer $t$. With this interpretation, the objective captures the expected revenue. Furthermore, the first (second) set of constraints captures a \emph{relaxed} inventory feasibility constraint for all units in the first batch (second and later batches) of each product.\footnote{Fixing a product, the first constraint is an aggregated capacity constraint for all the units in the initial ready batch and the batch right after that (imagine those units are merged with the first batch). Therefore, the right-hand side is $\inventory_i+\gamma$. Then the constraint corresponding to each index $j\in[2:h_i]$ is a capacity constraint for units in batch $j+1$ (imagine that those units are shifted to batch $j$, as if they have arrived earlier). As shifting units to arrive earlier makes the problem more relaxed, this program is a relaxation. See the proof of \Cref{lem:relaxation} for a formal argument.}
The third set of constraints ensures that the sum of allocation probabilities is at most one for each consumer.

As the main technical lemma in this subsection, we formally show that for any execution of the BIB algorithm, its induced program~\ref{eq:opt} is an upper bound on the revenue of the clairvoyant optimum policy that knows the adversarial sequence of consumer choice models and amounts of exogenous inventory shocks (and hence a relaxation to any feasible policy): 
\begin{restatable}{lemma}{lemrelax}
\label{lem:relaxation}
For any sequence $\{\choice_t, \repleni\}_{i\in[n],t\in[\totaltime]}$,
the expected revenue of the clairvoyant optimum policy (and hence any online algorithm) is at most the optimal objective value in program~\ref{eq:opt} with any choice of batch-size threshold $\batchsizethreshold\in\naturals$.
\end{restatable}
\begin{proof}{\textsl{Proof.}}
Given any $\batchsizethreshold\in\naturals$, we partition all units of each product into batches based on the batch construction/assignment of the BIB algorithm (with batch-size threshold $\batchsizethreshold$) and equivalently consider the batch-specified assortment as we defined above. In the following, we construct a feasible solution for program~\ref{eq:opt} using the execution path of the clairvoyant optimum policy. 

\revcolor{
For each time period $t$, let $\tilde{\indexnew}_t$ denote the batch-specified assortment offered by the clairvoyant policy, where $\tilde{\indexnew}_t$ is a sample-path dependent random subset. 
For each unit allocation of product $i$ in batch $\batch_{i,j}$ by the clairvoyant policy, we reassign it to the earlier batch $\batch_{i,j-1}$ (with $\batch_{i,0}$ interpreted as $\batch_{i,1}$). 
Let $\tilde{\indexnew}\primed_t$ denote the batch-specified assortment obtained by applying this batch-shifting operation to $\tilde{\indexnew}_t$, i.e., for each $(i, j) \in \tilde{\indexnew}_t$, include $(i, j - 1\vee 1)$ in $\tilde{\indexnew}\primed_t$. 
We then set $\alloc_{\tilde{\indexnew}\primed_t,t}=1$ for all $t$, and $\alloc_{\indexnew,t}=0$ for all other $\indexnew \in \indexset_t$. We next verify the constructed solution is feasible and finally analyze its objective value. 


\noindent\emph{First set of constraints.}
Consider any product $i$ and any interval of length $\duration_i$. 
The clairvoyant policy can allocate at most $\inventory_i$ units from the initial inventory plus those from the first exogenous batch. 
By construction, such allocations are shifted to $\batch_{i,1}$ in our constructed solution. 
Since $\batch_{i,1}$ has total size $\inventory_i+\batchsizethreshold$, the number of shifted allocations never exceeds this amount. 
Hence, the first family of constraints in \ref{eq:opt} is satisfied.

\noindent\emph{Second set of constraints.}
For each $j \ge 2$, each allocation of a unit from $\batch_{i,j+1}$ in the clairvoyant policy is shifted to $\batch_{i,j}$ in our constructed solution. 
As each batch has size at most $\batchsizethreshold$, in any interval of length $\duration_i$ there can be at most $\batchsizethreshold$ such effective allocations. 
Thus, the second family of constraints in \ref{eq:opt} holds.

\noindent\emph{Third set of constraints.}
For each time $t$, we set $\alloc_{\tilde{\indexnew}\primed_t,t}=1$ and all other allocations at time $t$ to zero. 
Therefore, $\sum_{\indexnew \in \indexset_t} \alloc_{\indexnew,t}=1$ for each $t$, which satisfies the third family of constraints.

\noindent\emph{Objective value.}
Finally, note that for every period $t$, the shifted assortment $\tilde{\indexnew}\primed_t$ yields the same choice probabilities as the clairvoyant assortment $\tilde{\indexnew}_t$. 
Therefore, the expected objective value of our constructed solution matches that of the clairvoyant policy.}
%
%
\hfill\halmos
\end{proof}

%% file: Paper/primal-dual-analysis.tex
\label{sec:IBB proof}

In this section, we present the competitive ratio analysis of the BIB algorithm using a primal-dual framework. 
\revcolor{Before delving into analysis, we introduce a key combinatorial construct underpinning our dual construction.

\xhdr{Interval assignment problem (IAP).}
Consider a sequence of $s$ intervals indexed by $[s] \triangleq \{1, 2, \dots, s\}$, denoted $\{[a_i, b_i]\}_{i \in [s]}$, whose left endpoints satisfy $a_1 < a_2 < \dots < a_s$. For each interval index $i \in [s]$, define the \emph{coverage set}
$\contain_i = \{j \in [s] : a_i \in [a_j, b_j]\}$,
i.e., the set of interval indices whose intervals cover the left endpoint $a_i$, and let $|\contain_i|$ denote its cardinality.
The \emph{Interval Assignment Problem (IAP)} asks to assign a positive integer $\assign_i \in \naturals$ to each interval index $i$ such that the resulting assignment $(\assign_1, \dots, \assign_s)$ satisfies the following three structural properties:
\begin{enumerate}[label=(\roman*)]
    \item \textsl{\underline{(\LocalDominance):}} 
    For each interval index $i \in [s]$, there exists a bijection $\rho_i : \contain_i \to \{1, 2, \dots, \containni\}$ such that $\assign_j \geq \rho_i(j)$ for all $j \in \contain_i$.
    
    \item \textsl{\underline{(\GlobalDominance):}} 
    For every decreasing and concave function $\bar\pen$, we have
    $\sum_{i \in [s]} \bar\pen(\assign_i) \geq \sum_{i \in [s]} \bar\pen(\containni).$
    
    \item \textsl{\underline{(\PartitionMonotonicity):}} 
    The interval index set can be partitioned into disjoint subgroups $A_1, \dots, A_\kappa$ such that for each $\ell \in [\kappa]$,
    $\{\assign_i : i \in A_\ell\} = \{1, 2, \dots, |A_\ell|\}.$
\end{enumerate}
We now state our main result about the IAP: a polynomial-time algorithm exists that always produces a valid assignment satisfying all three properties. Illustrative examples for the IAP can be found in \Cref{apx:IAP illustration}.
\begin{lemma}
\label{lem:IAP}
For any input instance of the IAP as defined above, there exists a polynomial-time algorithm (\Cref{alg:interval assignment}; see \Cref{sec:IAP} for details) that computes an assignment $(\assign_1, \dots, \assign_s) \in \naturals^s$ satisfying \hyperref[lem:IAP]{\LocalDominance}, \hyperref[lem:IAP]{\GlobalDominance}, and \hyperref[lem:IAP]{\PartitionMonotonicity}. 
\end{lemma}
}

\xhdr{Overview of the solution to IAP and its applications.}
\revcolor{In the primal-dual analysis of the BIB algorithm, a central difficulty arises because the actual number of units that have been allocated from a batch—but not yet returned—may increase and decrease over time due to overlapping rental periods. This non-monotonicity prevents a direct application of standard dual constructions, which rely on inventory trajectories that only decrease.\footnote{\revcolor{cf.\ \Cref{example:failure of actual consumption level}, where the naive assignment based on real consumption fails.}}
To resolve this, for each ``ready'' batch, we construct a \emph{transformed} inventory trajectory by solving an instance of IAP. In this instance, each allocation from the batch corresponds to an interval spanning from its allocation time to its return time. The IAP maps these intervals to integer labels that serve as surrogate inventory levels. These labels are carefully designed to satisfy three key properties, each of which enables a critical step in the primal-dual analysis:
\hyperref[lem:IAP]{\LocalDominance} ensures that, at the moment any unit is allocated, the transformed inventory levels of all units currently in use are sufficiently large to support the dual feasibility condition. Intuitively, it guarantees that the dual variables can be set high enough to ``cover'' the primal revenue of the allocation.
\hyperref[lem:IAP]{\GlobalDominance} guarantees that the total penalty (i.e., the sum of $\Psi$ evaluated at the transformed inventory levels) is no worse than if we had used the real, non-monotone consumption levels. This allows us to relate the algorithm's actual revenue to the dual objective without loss.
\hyperref[lem:IAP]{\PartitionMonotonicity} reorganizes the allocations into groups where the transformed inventory levels behave like those in the classic setting—strictly decreasing by one after each allocation. This structure is essential for recovering the tight competitive ratio for general penalty function $\pen$ in \citet{GNR-14}.}
The details of how to solve IAP and the proof of \Cref{lem:IAP} are deferred to \Cref{sec:IAP}.


We now present our primal-dual analysis for the BIB algorithm using the key \Cref{lem:IAP}.
\begin{proof}{\textsl{Proof of \Cref{thm:concave competitive ratio}.}}
We start our analysis by writing down the dual program of our LP benchmark~\ref{eq:opt}: 
\begin{align*}
\begin{array}{llll}
\min
\limits_{\boldsymbol{\theta},\boldsymbol{\lambda} \geq \mathbf 0} &
\displaystyle\sum\nolimits_{i=1}^n
\left(
\displaystyle\sum\nolimits_{t=\addtime_{i,1}}^\totaltime
(\inventory_i + \batchsizethreshold)\inventorydual_{i,1,t}
+
\displaystyle\sum\nolimits_{j=2}^{\cnt_i}
\displaystyle\sum\nolimits_{t=\addtimeij}^\totaltime
\batchsizethreshold\inventorydualijt
\right)
+
\displaystyle\sum\nolimits_{t =1}^\totaltime
\probduali
\quad\quad
& 
\text{s.t.} 
&\\ 
&
\probduali + 
\displaystyle\sum\nolimits_{(i,j)\in \indexnew}
\displaystyle\sum\nolimits_{\tpre= t - \duration_i + 1}^t
\choice_t(\indexnew,i)\inventorydual_{i,j,\tpre} 
\geq \displaystyle\sum\nolimits_{(i,j)\in \indexnew}
\reward_i\choice_t(\indexnew, i)
& t \in [\totaltime],\ \indexnew \in \indexsett
&~.
\end{array}
\end{align*}
We construct a randomized dual assignment based on the execution path of the BIB algorithm together with the random realization of consumers' choices as follows. First, set $\probduali \gets 0$ and $\inventorydualijt \gets 0$ for all $t\in[T],i\in[n],j\in[\cnt_i]$. For each time period $t\in[T]$, suppose consumer $t$ chooses a unit of product $i$ from batch $\batchij$, update dual variables as follows:
\begin{align*}
    \probduali \gets \reward_i\pen\left(\inventoryratioijt\right)
    \qquad\mbox{and}\qquad
    \inventorydualijt \gets \reward_i \left(\pen\left(\pseudoinventoryratioijt\right)- \pen\left(\pseudoinventoryratioijt-({1}/{|\batchij|})\right)\right)~.
\end{align*}
Here \revcolor{$\{\pseudoinventoryratioijt\}_{i\in[n],j\in[\cnt_i],t\in[T]}$} is a sequence of values constructed based on the interval assignment problem (IAP) (\Cref{lem:IAP}). Specifically, for each product $i\in[n]$ and its ``ready'' batch $\batchij$ with $j\in[\cnt_i]$, let $\timesetij = \{t_1, t_2, \dots, t_s\}$ be the time periods when a unit of product $i$ from $\batchij$ is chosen by consumers. Consider an instance of IAP with $s$ intervals $[t_1, t_1 + \duration_i], \dots, [t_s, t_s + \duration_i]\}$. Let $\{\assign_{k}\}_{k\in [s]}$ be the assignment solved in IAP for this batch. With slight abuse of notation, we also denote the assignment solved in this IAP instance as $\{\assign_{i,j,t}\}_{t\in\timesetij}$ to emphasize its connection to the BIB algorithm's execution. We set $\pseudoinventoryratio_{i,j,t} \triangleq 1 - {(\assign_{i,j,t} - 1)}/{|\batchij|}$ for every time $t\in\timesetij$ and $\pseudoinventoryratioijt \triangleq 0$ for all the remaining time $t\in[T]\backslash\timesetij$.


\revcolor{The analysis proceeds in two steps. First, we show that the randomized dual assignments are feasible in expectation. Second, we show that the expected dual objective provides a $\competitiveratio(\pen,\batchsizethreshold)$ approximation to the expected revenue of the BIB algorithm. Combining these two results with the strong duality establishes the competitive ratio of the algorithm, as stated in \Cref{thm:concave competitive ratio}.}

\noindent
[\emph{Step i}] \emph{Checking the feasibility of dual in expectation.}
Since penalty function $\pen(\cdot)$ of the BIB algorithm is non-negative and monotone non-decreasing, the constructed dual assignment is non-negative. We next show that the dual constraints corresponding to each $\curalloc$ is satisfied in expectation. Note that this is enough, because we can later use the expectation of the dual variables as the actual dual certificate. First, we establish the following inequality for every product $i\in[n]$, batch $j\in[\cnt_i]$, and time $t\in[\addtimeij:T]$:
\begin{align}
    \label{eq:dual feasible ex post general}
    \reward_i
    \pen\left(\inventoryratioijt\right) 
    + \displaystyle\sum\nolimits_{\tpre=t-\duration_i + 1}^{t - 1}
    \inventorydual_{i, j, \tpre} 
    \geq \reward_i~.
\end{align}
To prove inequality~\eqref{eq:dual feasible ex post general}, consider on-hand units together with units of product $i$ that have not yet returned in batch $\batchij$ at time $t$. Since the normalized inventory level of batch $\batchij$ at time $t$ is $\inventoryratioijt$, there are $m \triangleq (1 - \inventoryratioijt) |\batchij|$ units of product $i$ in this batch that have not yet returned at time $t$. Suppose that these units were sold at times $\tpre_1\leq\tpre_2\leq \dots\leq\tpre_{m}$. By construction, all these time periods are between $t-\duration_i + 1$ and $t - 1$ and belong to $\timesetij$. We now show inequality~\eqref{eq:dual feasible ex post general} as follows:
\begin{align*}
    &\reward_i
    \pen\left(\inventoryratioijt\right) 
    + \sum\nolimits_{\tpre=t-\duration_i + 1}^{t - 1}
    \inventorydual_{i, j, \tpre} \overset{(a)}{=} 
    \reward_i
    \pen\left(\inventoryratioijt\right) 
    + \sum\nolimits_{k = 1}^m
    \inventorydual_{i, j, \tpre_k} 
    \\
    {}\overset{(b)}{=} {}&
    \reward_i
    \pen\left(\inventoryratioijt\right) 
    + \sum\nolimits_{k = 1}^m
    \reward_i\left(
    \pen\left(\pseudoinventoryratio_{i,j,\tpre_k}\right)
    -
    \pen\left(\pseudoinventoryratio_{i,j,\tpre_k}-({1}/{|\batchij|})\right)
    \right) 
    \\
    {}\overset{(c)}{=} {}&
    \reward_i
    \pen\left(\inventoryratioijt\right) 
    + \sum\nolimits_{k = 1}^m
    \reward_i\left(
    \pen\left(
    1-((\assign_{i,j,\tpre_k} - 1)/|\batchij|)
    \right)
    -
    \pen\left(
    1 - (\assign_{i,j,\tpre_k}/|\batchij|)
    \right)
    \right) 
    \\
    {}\overset{(d)}{\geq} {}&
    \reward_i
    \pen\left(
    1 - (m/|\batchij|)
    \right) 
    + \sum\nolimits_{k = 1}^m
    \reward_i\left(
    \pen\left(
    1 - ((k-1)/|\batchij|)
    \right)
    -
    \pen\left(
    1 - (k/|\batchij|)
    \right)
    \right) 
    \overset{(e)}{=} 
    \reward_i\pen(1) 
    \overset{(f)}{=} \reward_i~.
\end{align*}
Here, equality (a) holds due to the definition of $m$, $\{\tpre_k\}_{k\in[m]}$ and the dual assignment construction; Equality~(b) holds due to the dual construction; Equality~(c) holds due to the construction of $\{\pseudoinventoryratio_{i,j,\tpre_k}\}$. For inequality~(d), we use the definition of $m$, and the rest of inequality holds due to \hyperref[lem:IAP]{\LocalDominance} of IAP (\Cref{lem:IAP}) at time $\tpre_m$, and the concavity of the penalty function $\pen$.\footnote{\label{footnote:LD usage explanation}\revcolor{Let $m = (1 - \inventoryratioijt)|\batchij|$ be the number of units from batch $\batchij$ in use at time $t$, allocated at times $\tpre_1 < \dots < \tpre_m$. Consider the last such time $\tpre_m$; the coverage set $\contain_{\tpre_m}$ includes all $m$ intervals, so $|\contain_{\tpre_m}| = m$. By \hyperref[lem:IAP]{\LocalDominance}, the IAP assignment satisfy $\assign_{(k)} \geq k$ for all $k \in [m]$, where $\assign_{(k)}$ is the $k$-th smallest label.
Since $\pen$ is concave and non-decreasing, its marginal difference is non-increasing. Thus, replacing the IAP assignment $\assign_{i,j,\tpre_k}$ with the smaller values $k$ can only increase each marginal term. Consequently, the sum on the LHS—using the IAP assignment—is at least as large as the RHS—using the idealized sequence $k = 1,\dots,m$. This justifies inequality~(d).}} Finally, equality~(e) holds due to the telescopic sum and equality~(f) holds since $\pen(1) = 1$.

Fixing an arbitrary time period $t$ and a feasible batch-specified assortment $\indexnew \in \indexsett$, we establish the following inequality for every time period $t$, which takes the expectation over the randomized choice of the consumer $t$', while conditioning on the realized sample path $\mathcal{F}_{t-1}$ up to the end of time period $t$ and the realization of the batch-specified assortment $\indexnew_t$ displayed to consumer $t$:
\begin{align}
\label{eq:lambda conditional lowerbound general} 
\begin{split}
\expect{
		\probduali  \condition\mathcal{F}_{t-1},\indexnew_t
		}
		\overset{(a)}{=}
		\displaystyle\sum\nolimits_{(i,j)\in 
		\indexnew_t}
		\reward_i\choice_{t}(\indexnew_t, i)
		\pen\left(\inventoryratioijt\right)
		\overset{(b)}{\geq} 
		\displaystyle\sum\nolimits_{(i,j)\in 
		\indexnew}
		\reward_i\choice_{t}(\indexnew , i)
		\pen\left(
		\inventoryratioijt
		\right)~.
		\end{split}
\end{align}
Here, equality~(a) holds due to the dual assignment construction and the expectation over consumer $t$'s choice; and inequality~(b) holds due to the greedy decision for selecting $\indexnew_t$ in BIB.

Combining inequalities~\eqref{eq:dual feasible ex post general} and \eqref{eq:lambda conditional lowerbound general}, we are able to verify the feasibility of the dual constraint corresponding to primal variable $\curalloc$ as follows:
\begin{align*}
    &\expect{
	\probduali + 
		\displaystyle\sum\nolimits_{(i,j)\in \indexnew}
		\displaystyle\sum\nolimits_{\tpre= t - \duration_i + 1}^t
		\choice_t(\indexnew,i)\theta_{i,j,\tpre}
		}
  \overset{(a)}{\geq}{}
    \expect{
		\displaystyle\sum\nolimits_{(i,j)\in 
		\indexnew}
		\reward_i\choice_{t}(\indexnew , i)
		\pen\left(
		\inventoryratioijt
		\right) + 
		\displaystyle\sum\nolimits_{(i,j)\in \indexnew}
		\displaystyle\sum\nolimits_{\tpre= t - \duration_i + 1}^t
		\choice_t(\indexnew,i)\theta_{i,j,\tpre}
		}
  \\
  \overset{(b)}{\geq}{}&
  \displaystyle\sum\nolimits_{(i,j)\in 
		\indexnew}
		\choice_{t}(\indexnew , i)
    \expect{
		\reward_i
		\pen\left(
		\inventoryratioijt
		\right) + 
		\displaystyle\sum\nolimits_{\tpre= t - \duration_i + 1}^t
		\theta_{i,j,\tpre}
		}
  \overset{(c)}{\geq}
  \displaystyle\sum\nolimits_{(i,j)\in 
		\indexnew}
  \reward_i
		\choice_{t}(\indexnew , i)~,
\end{align*}
where inequality~(a) holds due to inequality~\eqref{eq:lambda conditional lowerbound general}; inequality~(b) holds due to the linearity of expectation; and inequality~(c) holds due to inequality~\eqref{eq:dual feasible ex post general}.

\smallskip
\noindent
[\emph{Step ii}] \emph{Comparing objective values of primal and dual.}
We compare the contribution of batch $\batchij$ in the BIB algorithm and the contribution in the dual objective for each ``ready'' batch $\batchij$ separately. Recall $\timesetij$ is the set of time periods when a unit of product $i$ from batch $\batchij$ has been selected. The revenue contribution in the BIB algorithm (aka., primal) corresponding to batch $\batchij$ can be expressed as
$\text{Primal}_{i,j}  = \reward_i|\timesetij|$.
Meanwhile, the contribution of the dual objective based on the constructed dual assignment can be upper-bounded as
\begin{align}
\begin{split}
\label{eq:dual obj upper bound}
\text{Dual}_{i,j} 
&{}\overset{(a)}{=} 
\reward_i
\displaystyle\sum\nolimits_{t \in \timesetij} 
\left(
\pen\left(\inventoryratioijt\right) 
+
\left(\inventory_i \cdot \indicator{j=1} + \batchsizethreshold\right)\left(
\pen\left(\pseudoinventoryratioijt\right) 
-
\pen\left(\pseudoinventoryratioijt - 
(1/|\batchij|)
\right) \right)
\right)
\\
&{}\overset{(b)}{\leq} 
\reward_i
\displaystyle\sum\nolimits_{t \in \timesetij} 
\left(
\pen\left(\pseudoinventoryratioijt\right) 
+
\left(\inventory_i \cdot \indicator{j=1} + \batchsizethreshold\right)\left(
\pen\left(\pseudoinventoryratioijt\right) 
-
\pen\left(\pseudoinventoryratioijt - 
(1/|\batchij|)
\right) \right)
\right)~,
\end{split}
\end{align}
where equality~(a) holds due to the constructed dual assignment; and inequality~(b) holds due to the construction of $\{\pseudoinventoryratioijt\}$ and \hyperref[lem:IAP]{\GlobalDominance} of IAP (\Cref{lem:IAP}).\footnote{\label{footnote:GD usage explanation}\revcolor{Recall that $\inventoryratioijt = 1 - m_t / |\batchij|$, where $m_t$ is the number of units from batch $\batchij$ in use at time $t$, 
and $\pseudoinventoryratioijt = 1 - (\assign_{i,j,t} - 1)/|\batchij|$, where $\assign_{i,j,t}$ is the IAP assignment. Define $\bar\pen(x) \triangleq \pen(1 - (x - 1)/|\batchij|)$, which is decreasing and concave because $\pen$ is concave and non-decreasing. 
Since $m_t + 1 = |\contain_t|$ (the coverage size at the allocation time $t$), we have $\pen(\inventoryratioijt) = \bar{\pen}(|\contain_t|)$ and $\pen(\pseudoinventoryratioijt) = \bar{\pen}(\assign_{i,j,t})$. 
\hyperref[lem:IAP]{\GlobalDominance} then implies $\sum_{t \in \timesetij} \bar{\pen}(\assign_{i,j,t}) \geq \sum_{t \in \timesetij} \bar{\Psi}(|\mathcal{N}_t|)$, 
i.e., $\sum_{t \in \timesetij} \Psi(\pseudoinventoryratioijt) \geq \sum_{t \in \timesetij} \pen(\inventoryratioijt)$, which justifies inequality~(b).
}} 

Now, we partition time periods in $\timesetij$ into subgroups $A_{i_,j, 1}, A_{i,j, 2}, \dots, A_{i,j, \kappa}$ induced by \hyperref[lem:IAP]{\PartitionMonotonicity} of IAP (\Cref{lem:IAP}). Our remaining analysis compares the contribution of batch $\batchij$ from time periods $A_{i,j,\ell}$ in the BIB algorithm (denoted by $\text{Primal}_{i,j,\ell}$) and the corresponding dual assignments' contribution in the dual objective (denoted by $\text{Dual}_{i,j,\ell}$) for each $\ell\in[\kappa]$ separately. 
Define auxiliary notation $\tilde{\inventory}_{i,j} \triangleq \inventory_i \cdot \indicator{j=1} + \batchsizethreshold$.
Note that
$\text{Primal}_{i,j,\ell} = \reward_i|A_{i,j,\ell}|$
and
\begin{align*}
\text{Dual}_{i,j,\ell} 
&{}\overset{(a)}{\leq } \reward_i
\displaystyle\sum\nolimits_{t \in A_{i,j,\ell}} 
\left(
\pen\left(\pseudoinventoryratioijt\right) 
+
\tilde{\inventory}_{i,j} 
\left(
\pen\left(\pseudoinventoryratioijt\right) 
-
\pen\left(\pseudoinventoryratioijt - 
(1/|\batchij|)
\right) \right)
\right)
\\
&{}\overset{(b)}{\leq } \reward_i
\displaystyle\sum\nolimits_{t \in A_{i,j,\ell}} 
\left(
\pen\left(1-
((\assign_{i,j,t}-1)/|\batchij|)
\right) 
+
\tilde{\inventory}_{i,j} 
\left(
\pen\left(1-
((\assign_{i,j,t}-1)/|\batchij|)
\right) 
-
\pen\left(1-
({\assign_{i,j,t}}/{|\batchij|})
\right) \right)
\right)
\\
&{}\overset{(c)}{=} \reward_i
\displaystyle\sum\nolimits_{k=1}^{|A_{i,j,\ell}|} 
\left(
\pen\left(1-
((k-1)/|\batchij|)
\right) 
+
\tilde{\inventory}_{i,j} 
\left(
\pen\left(1-
((k-1)/|\batchij|)
\right) 
-
\pen\left(1-
(k/|\batchij|)
\right) \right)
\right)
\\
&{}\overset{}{=} \reward_i
\displaystyle\sum\nolimits_{k=|\batchij| - |A_{i,j,\ell}|+1}^{|\batchij|} 
\left(
\pen\left(
k/|\batchij|
\right) 
+
\tilde{\inventory}_{i,j} 
\left(
\pen\left(
k/|\batchij|
\right) 
-
\pen\left(
(k-1)/|\batchij|
\right) \right)
\right)
\\
&{}\overset{(d)}{=} \reward_i
\displaystyle\sum\nolimits_{k=|\batchij| - |A_{i,j,\ell}|+1}^{|\batchij|} 
\pen\left(
k/|\batchij|
\right) 
+
\reward_i
\tilde{\inventory}_{i,j} 
\left(
\pen\left(1\right) 
-
\pen\left(
1-(|A_{i,j,\ell}|/|\batchij|)
\right)
\right)
\\
&{}\overset{(e)}{\leq} \reward_i + 
\reward_i
\displaystyle\int\nolimits_{
1 - ((|A_{i,j,\ell}|-1)/|\batchij|)
}^{1} 
\pen\left(y\right) \,\dd y
+
\reward_i
\tilde{\inventory}_{i,j} 
\left(
1
-
\pen\left(
1 - (|A_{i,j,\ell}|/|\batchij|)
\right)
\right)
\end{align*}
where inequality~(a) holds due to inequality~\eqref{eq:dual obj upper bound}; inequality~(b) holds due to the construction of $\pseudoinventoryratioijt$; equality~(c) holds due to \hyperref[lem:IAP]{\PartitionMonotonicity} of IAP (\Cref{lem:IAP}); equality~(d) holds due to the telescopic sum; and inequality~(e) holds since  $\pen$ is weakly increasing and $\pen(1) = 1$.

Note that by setting $x \triangleq 1-
(|A_{i,j,\ell}|/{|\batchij|})$, the ratio between $\text{Primal}_{i,j,\ell}$ and $\text{Dual}_{i,j,\ell}$ has the same closed form as $\competitiveratio_1(\pen,\batchsizethreshold)$ for ``ready'' batch $\batch_{i,1}$,
and $\competitiveratio_2(\pen,\batchsizethreshold)$ for ``ready'' batch $\batchij$ with $j \in [2:\cnt_i]$. Thus, we conclude that for every sample path, the revenue in the BIB algorithm is a $\competitiveratio(\pen,\batchsizethreshold)$ approximation of the dual objective value of the constructed dual assignment.
\hfill\halmos
\end{proof}

%% file: Paper/interval-assignment.tex
\label{sec:IAP}

\revcolor{We now prove \Cref{lem:IAP} by explicitly constructing an assignment that satisfies all three required properties. We first describe its core idea and then verify the properties using its structural invariants. A reference implementation is provided in \Cref{alg:interval assignment}, and illustrative executions (\Cref{fig:failure of actual consumption level,fig:failure of greedy}) of the algorithm on representative IAP instances (\Cref{example:failure of actual consumption level,example:failure of greedy}) can be found in \Cref{apx:IAP illustration}. 

\xhdr{Algorithmic idea.}
The assignment is constructed by processing intervals in a specific order dictated by their \emph{current coverage size}: at each step, among all unassigned intervals (denoted by $\remaining\subseteq [s]$), we identify the one (denoted by $i$) that is covered by the \emph{largest} number of still-unassigned intervals (i.e., $\contain_i\cap\remaining$), and assign the \emph{leftmost} interval $j$ in that coverage set a label equal to the current coverage size $|\contain_i\cap\remaining|$). This greedy rule ensures that intervals appearing earlier (with smaller indices) are never ``starved'' of large labels, which is essential for \hyperref[lem:IAP]{\LocalDominance}. Moreover, because the algorithm always assigns labels based on maximal coverage sets, the resulting assignment incurs no more penalty (under any decreasing concave function $\bar\pen(\cdot)$) than the real consumption levels, i.e., satisfies \hyperref[lem:IAP]{\GlobalDominance}.

To enable \hyperref[lem:IAP]{\PartitionMonotonicity}, the algorithm also builds a predecessor mapping during execution. Specifically: (i) mapping $\pmother_j$ records which interval's coverage set triggered the assignment of $\assign_j$. This is used to determine which intervals overlap in a way that justifies linking them into a partition group; (ii) mapping $\pfather_k$ records a predecessor link: when interval $j$ is assigned label $m$, it becomes the predecessor of any previously assigned interval $k$ with $\assign_k= m+1$ that overlaps appropriately with $j$'s coverage context (i.e., $j \in \contain_{\pmother_k}$). These links induce a forest of chains, each forming a valid partition group.
These mappings are not used during assignment decisions—they exist solely to support the analysis.}

\input{Paper/alg-IAP}

\revcolor{We now prove \Cref{lem:IAP} by formally verifying the three properties using our constructed procedure.}

\begin{proof}{\textsl{Proof of \Cref{lem:IAP}}.}
\revcolor{The constructed procedure (\Cref{alg:interval assignment}) terminates after $s$ steps, producing a complete assignment $(\assign_1, \dots, \assign_s)$.} We next verify that this assignment satisfies all three desired properties in the statement of the lemma.

\noindent\emph{(\rom{1}) \hyperref[lem:IAP]{\LocalDominance}:}
For each interval $i \in [s]$, consider $j_1, j_2, \dots \in \contain_i$ in the order in which they are removed from the algorithm: for the $k$-th removed interval $j_k \in \contain_i$, at that time the cardinality $\remainingcontainni$ is $\containni - k + 1$, which guarantees that $\assign_{j_k}$ is at least $\containni - k + 1$ by construction. Thus, \hyperref[lem:IAP]{\LocalDominance} is satisfied.

\noindent\emph{(\rom{2}) \hyperref[lem:IAP]{\GlobalDominance}:}
We verify this property by the following induction argument. Consider the reverse order of intervals removed from the remaining interval set $\remaining$ in the algorithm, i.e., we set $\remaining = \emptyset$ as the base case and then add intervals back to $\remaining$ in the reverse order in which they have been removed in the algorithm. Our induction hypothesis is that restricting to intervals in $\remaining$, \hyperref[lem:IAP]{\GlobalDominance} holds, i.e., for any decreasing and concave function $\bar\pen$, 
\begin{align}
\label{eq:global dominance}
\sum\nolimits_{i\in \remaining} \bar\pen(\assign_i) \geq \sum\nolimits_{i \in \remaining} \bar\pen(\remainingcontainni).
\end{align}

\noindent\textbullet\ \emph{Base Case ($\remaining = \emptyset$):} In this case, the induction hypothesis holds trivially.

\noindent\textbullet\ \emph{Inductive Step:}  Suppose the induction hypothesis holds before we add interval $i$ with assignment~$\assign_{i}$ back to $\remaining$. Let $m \triangleq \assign_{i} - 1$. To make it clear, let $\remaining$ to be the remaining interval set before we add interval $i$, and $\remaining'$ to be the remaining interval set after we add interval $i$, i.e., $\remaining' = \remaining \cup\{i\}$. Next we analyze the increment of both sides in inequality~\eqref{eq:global dominance} after we add interval $i$ into $\remaining$. The left-hand side of inequality~\eqref{eq:global dominance} increases by $\bar\pen(m + 1)$. For the right-hand side, since $\assign_i = m + 1$, there exist $m$ intervals $j_1, j_2, \dots j_m \in \remaining$ which cover position $a_{\pmother_i}$, and $i < j_k$ (otherwise, $j_k$ should be removed earlier than $i$ in the algorithm) for all $k\in[m]$. This implies that interval $i$ covers position $a_{j_k}$ for all $k \in[m]$. Without loss of generality, we assume $j_1 < j_2 < \dots < j_m$. We have $|\remainingcontain_{j_k}|\geq k$ for all $k\in[m]$. Hence, after adding interval $i$ into $\remaining$, the right-hand side of inequality~\eqref{eq:global dominance} increases by 
\begin{align*} 
&\bar\pen(|\contain_i^{\remaining'}|) + 
\sum\nolimits_{j\in \remaining}
\left(
\bar\pen(|\contain_{j}^{\remaining'}|) - 
\bar\pen(|\contain_{j}^{\remaining}|)
\right) 
\overset{(a)}{\leq} \bar\pen(|\contain_i^{\remaining'}|) + 
\sum\nolimits_{k\in[m]}
\left(
\bar\pen(|\contain_{j_k}^{\remaining'}|) - 
\bar\pen(|\remainingcontain_{j_k}|)
\right) \\
\overset{(b)}{=}{}&
\bar\pen(|\contain_i^{\remaining'}|) + 
\sum\nolimits_{k\in[m]}
\left(\bar\pen(|\remainingcontain_{j_k}| + 1) - 
\bar\pen(|\remainingcontain_{j_k}|)\right) 
\overset{(c)}{\leq} 
\bar\pen(1) + 
\sum\nolimits_{k \in [m]}
(
\bar\pen(k + 1) - \bar\pen(k)
)
= \bar\pen(m + 1)
\end{align*}
where inequality~(a) holds since $|\contain_{j_k}^{\remaining'}| \geq |\contain_{j_k}^{\remaining}|$ and $\bar\pen$ is decreasing; equality~(b) holds since $|\contain_{j_k}^{\remaining'}|  = |\remainingcontain_{j_k}| + 1$ by construction; and inequality~(c) holds since $\bar\pen$ is decreasing, concave and $|\remainingcontain_{j_k}|\geq k$ for all $k\in[m]$. Applying the induction hypothesis (before we add interval $i$), the inductive statement (after we add interval $i$) holds, which finishes the induction as desired.

\noindent\emph{(\rom{3}) \hyperref[lem:IAP]{\PartitionMonotonicity}:}
\revcolor{We first show that the predecessor mapping $\{\pfather_i\}_{i \in [s]}$ satisfies two key properties:  
(a) each interval has at most one successor $j$ (i.e., $\pfather_j = i$), and  
(b) every interval $i$ with $\assign_i > 1$ has a predecessor (i.e., $\pfather_i \in [s]$).  
Consequently, the $\pfather$-links decompose the set of intervals into disjoint chains along which the assigned labels decrease by exactly one. This yields a valid partition $\{A_1, \dots, A_\kappa\}$ such that $\{\assign_i : i \in A_\ell\} = \{1, 2, \dots, |A_\ell|\}$ for each $\ell$, thereby guaranteeing \hyperref[lem:IAP]{\PartitionMonotonicity}. 
Next, we establish claims~(a) and~(b) separately.}

\xhdr{Proof of claim (a).} We show claim (a) by contradiction. Suppose there exist intervals $i$, $j_1$ and $j_2$ such that $\pfather_{j_1} = \pfather_{j_2} = i$. By the construction of \Cref{alg:interval assignment}, $\assign_{j_1} = \assign_{j_2} = \assign_i + 1$. Now, consider case by case.

\noindent\textbullet\ \emph{Case 1 ($i < j_1$ or $i < j_2$):} Without loss of generality, we assume $i < j_1$,
and the case $i < j_2$ is symmetric.
By the construction of \Cref{alg:interval assignment},
$\pfather_{j_1} = i$ implies 
$i \in \contain_{\pmother_{j_1}}$. 
This leads to a contradiction because if $i < j_1$, the interval $i$ should be removed earlier than the interval $j_1$ in \Cref{alg:interval assignment}
which implies $\assign_i \geq \assign_{j_1}$.

\noindent\textbullet\ 
\emph{Case2 ($j_1, j_2 < i$):}

\hspace{5mm}\noindent\textbullet\ 
\emph{Case 2a ($\pmother_{j_1} < \pmother_{j_2}$ or $\pmother_{j_2} < \pmother_{j_1}$):}
Without loss of generality, we assume $\pmother_{j_1} < \pmother_{j_2}$ ($\pmother_{j_2} < \pmother_{j_1}$ is similar):

\hspace{10mm}\noindent\textbullet\ \
\emph{Case 2a1 (interval $j_1$ is removed before interval $j_2$):}
In this case, consider the moment 
when \Cref{alg:interval assignment} removes $j_1$ due to interval $\pmother_{j_1}$.
By the construction of \Cref{alg:interval assignment}, since $j_1$ is removed before $j_2$,
it implies
$|\remainingcontain_{\pmother_{j_1}}| > |\remainingcontain_{\pmother_{j_2}}|$.
Thus,
$\assign_{j_1}$ is strictly larger than $\assign_{j_2}$,
a contradiction.

\hspace{10mm}\noindent\textbullet\ \
\emph{Case 2a2 (interval $j_2$ is removed before interval $j_1$):}
In this case, consider the moment 
when \Cref{alg:interval assignment} removes $j_2$ due to interval $\pmother_{j_2}$.
Notice that interval $j_2$ also covers position $a_{\pmother_{j_1}}$,
i.e., $j_2 \in \contain_{\pmother_{j_1}}$.
By the construction of \Cref{alg:interval assignment}, since $j_2$ is removed before $j_1$,
it implies
$|\remainingcontain_{\pmother_{j_1}}|  \leq |\remainingcontain_{\pmother_{j_2}}|$.
Thus, once \Cref{alg:interval assignment} removes $j_2$, 
$|\contain_{\pmother_{j_1}}^{\remaining/\{j_2\}}|$ (which is an upper bound
of $\assign_{j_1}$) 
becomes strictly smaller than 
$|\remainingcontain_{\pmother_{j_2}}|$ (which is equal to $\assign_{j_2}$), 
a contradiction.

\hspace{5mm}\noindent\textbullet\ 
\emph{Case 2b ($\pmother_{j_1} = \pmother_{j_2}$):}
Without loss of generality, we assume interval $j_1$ is removed before interval $j_2$,
and the case where interval $j_2$ is removed before interval $j_1$ is 
symmetric.
Let $\pmother^* = \pmother_{j_1} = \pmother_{j_2}$. Consider the moment
when \Cref{alg:interval assignment} removes $j_1$ due to $\pmother^*$.
By the construction of \Cref{alg:interval assignment},
$\assign_{j_1} = 
|\remainingcontain_{\pmother^*}|$
and $\assign_{j_2} \leq |\contain_{\pmother^*}^{\remaining/\{j_1\}}|
= |\remainingcontain_{\pmother^*}| - 1$.
Thus, $\assign_{j_1}>\assign_{j_2}$,
which is a contradiction.

\xhdr{Proof of claim (b).}  
For any interval $i$ with $\assign_i > 1$, consider the next removed interval $j$ among $\contain_{\pmother_i}$ after \Cref{alg:interval assignment} removes interval $i$. It is sufficient to show $\assign_j = \assign_i - 1$, which implies that \Cref{alg:interval assignment} sets $\pfather_i = j$ when it removes interval $j$. Since $j$ is the next interval removed among $\contain_{\pmother_i}$ after interval $i$, we have $\assign_i - 1\leq \assign_j \leq \assign_i$. To show $\assign_j < \assign_i$, we use the proof by contradiction. Suppose $\assign_j = \assign_i$. If $i\in\contain_{\pmother_j}$, then at the moment when \Cref{alg:interval assignment} removes $i$, we have $\assign_i = |\contain_{\pmother_i}^{\remaining}|\geq|\contain_{\pmother_j}^{\remaining}| \geq \assign_j + 1 = \assign_i + 1$, which is a contradiction. If $i \not\in \contain_{\pmother_j}$, since $i < j$, we have mapping $\pmother_i < \pmother_j$ and at the moment when \Cref{alg:interval assignment} removes $i$, $|\contain_{\pmother_i}^\remaining| = \assign_i =\assign_j  \leq  \contain_{\pmother_j}^\remaining|$, a contradiction.
\hfill\halmos
\end{proof}

%% file: Paper/alg-IAP.tex
\begin{algorithm}
 	\caption{\textsc{Interval Assignment Solver
}}\label{alg:interval assignment}
 	\KwIn{sequence of $s$ intervals $\{[a_1, b_1], \dots, [a_s, b_s]\}$}
        \KwOut{assignment $(\assign_1, \dots, \assign_s)$}
 	
initialize assignment $\assign_i \gets 0$, mapping $\pfather_i \gets -1$, and mapping $\pmother_i \gets -1$ for all $i \in [s]$.

initialize remaining interval set $\remaining \gets \{1, 2, \dots, s\}$.
\tcc{define $\remainingcontain_i = \contain_i \cap \remaining$}

\While{$\remaining$ is not empty}{
 	
find interval $i$ with the largest index in remaining interval set $\remaining$ where $\remainingcontainni$ is maximized, i.e., $i \gets\, \argmax_{i\in \remaining} \remainingcontainni$. 

find interval $j$ with the smallest index in remaining interval set $\remaining$ which covers position $a_i$, i.e., $j \gets \min \remainingcontain_i$.

set assignment $\assign_j \gets \remainingcontainni$ and mapping $\pmother_j \gets i$.

\For{each interval $k \in [s] / \remaining$ such that 
$j \in \contain_{\pmother_k}$ and $\assign_k = \assign_j + 1$}{
set mapping $\pfather_k \gets j$. 
}

remove interval $j$ from remaining interval set $\remaining$ 
i.e., update
$\remaining \gets \remaining / \{j\}$.

}
\textbf{return} $\{\assign_1, \dots, \assign_s\}$
\end{algorithm}

%% file: Paper/conclusion.tex
\section{Conclusion}
\label{sec:conclusion}
\revcolor{\xhdr{Summary \& takeaways.}} We studied real-time assortment planning under both endogenous and exogenous inventory shocks. We introduced an extension of the inventory balancing algorithm for this problem, called Batched Inventory Balancing (BIB). This new algorithm keeps track of the exogenous inventory shock of each product separately by smartly batching units of this product into buckets in an adaptive fashion. It then penalizes the price of different units of the product using a concave penalty function with the (normalized) inventory of the bucket of the unit as input. We analyzed the BIB algorithm with general concave penalty functions. The analysis in the prior work for online assortment (without inventory shocks) breaks as inventory levels are not monotone and the function does not satisfy a certain differential equation. To handle this challenge, we considered an LP benchmark that depends on the execution of the BIB algorithm, designed a batch-based randomized dual-fitting primal-dual analysis, and most interestingly proposed a reduction from the correct randomized dual construction in our problem to another algorithmic yet combinatorial problem. This combinatorial problem, called Interval Assignment Problem (IAP), searches for a certain re-assignment of the bucket inventory levels, so that the new inventory levels are always monotone and can be used in a primal-dual analysis as before. We concluded by showing how to solve the IAP in polynomial time. We then show how to put all the pieces together and obtain a universal bound comparable to \cite{GNR-14}, but with both endogenous and exogenous inventory shock.

\revcolor{
\xhdr{Open problems \& directions.} There are several promising technical directions for future research stemming from our work, including: extending the model to stochastic or heterogeneous return durations; incorporating inventory holding or restocking costs; handling negative inventory shocks; studying joint inventory–assortment control; and exploring application-specific operational replenishment policies. A more detailed discussion of these open problems is provided in \Cref{apx:open-problems}.
}


%% file: Paper/related-work.tex

\label{apx:sec:further-related}
Online assortment
has an extensively growing 
literature
recently, both under the prior-free/adversarial and the Bayesian settings. For non-reusable products,
\citet{BKX-15} study a model with two products, i.i.d.\ consumer types and Poisson arrivals. \citet{CF-09} study a model with non-stationary consumer types. \citet{GNR-14} propose inventory balancing algorithms for assortment---inspired by the seminal work of \citet{MSVV-05} for online ad allocation---and analyze their performance guarantee in the adversarial setting using a primal-dual approach. This analysis is later improved and generalized to the case where there are finite steps in the penalty function~\cite{MS-19}.
For online assortment of reusable products, \citet{RST-17,BM-22,FNS-24} study the Bayesian setting under various modeling assumptions.
\citet{GGISUW-19,GIU-20} study the adversarial setting with stationary rental fees and rental duration distributions. 
\citet{SZZ-25} introduce a more general framework
in the adversarial setting 
with decaying rental fees and rental duration distributions.


\citet{vaz-23,udw-21} study a variant
of AdWords problem \citep[cf.][]{MSVV-05,meh-13},
where the budget (aka., capacity) 
of each offline node 
is initially unknown to the online algorithms,
and is only revealed to the algorithms
immediately after it is exhausted. 
They present the competitive ratio analysis 
for 
Perturbed Greedy algorithm 
and Generalized Perturbed Greedy algorithm.
\citet{MRSS-24}
introduce 
a new variant of the online
matching with stochastic reward \citep[cf.][]{MP-12}.
In their model, 
a fixed fraction of online nodes 
are single-minded (i.e., 
each of them
has a single incident edge).
They show the non-optimality of 
the Inventory Balancing (IB)
in this model, 
and propose a new modification
-- Adaptive Capacity (AC) -- 
with improved competitive ratio.
In all \citet{vaz-23,MRSS-24,udw-21},
the online algorithms suffer 
some specific forms of uncertainty 
about the capacities of offline nodes.
However, our model is fundamentally different 
from them and thus not directly comparable.

Our results are also comparable to the literature on online bipartite matching (with vertex arrival) and its extensions.
\citet{DFNSU-24} is the most related, where the authors consider the adversarial arrival setting without replenishment under small capacity (in particular, the capacity equals one) and restricts to problem instances with identical duration across offline nodes (aka products) and time. The paper proposes integral matching algorithms that achieve competitive ratios slightly better than 1/2, but still not optimal. In contrast, our paper allows for heterogeneous duration across offline nodes (aka products) and achieves an asymptotically optimal competitive ratio (rather than looking at the small capacity regime).
We refer the reader to related surveys and books \citep[cf.][]{meh-13,EIV-23} for a comprehensive study of the online bipartite matching literature.

As mentioned earlier, a major source of exogenous shocks to inventory is customer returns, either from the original channel or from other channels. The impact of customer returns on inventories and assortments has been studied in the literature, e.g. recent work of \cite{LJC-23,LJU-20,mia-23}. See \cite{BCDL-10} for an overview of
the older literature on inventory control with returns. In addition, we interpret endogenous shocks as replenishment of sold units of the product when the retailer employs a simple discrete-period restocking policy with a fixed lead time. This is in fact a simplification and there is a long literature in supply chain and inventory theory on various forms of replenishment policies; for example, see \cite{BCDL-10,BZJ-23} for a literature review. Simple restocking with fixed lead time is a special case of stochastic lead time, which has been studied in the literature, e.g. \cite{JR-04,MSW-15}. Also, a more practical restocking policy is base-stock with a lower-threshold~\citep{sim-78,SVV-20}, e.g., restocking only when $\alpha$ fraction or $k$ number of units is sold (instead of each time a unit is sold). We leave further investigations on these models for future research. 
\revcolor{Recently, \citet{KLU-25} extended the framework developed in this paper to more general settings (e.g., endogenous inventory shocks with i.i.d.\ stochastic durations). Their work builds directly on our core ideas and demonstrates a general reduction framework that lifts a broad class of algorithms—and their accompanying theoretical guarantees—from the setting without exogenous shocks to one that accommodates arbitrary exogenous inventory shocks.}

%% file: Paper/apx-practical.tex
\revcolor{
The online assortment planning model with endogenous and exogenous shocks studied in this paper is rather general and rooted in various applications. In this section, we discuss several such applications in detail and explain the micro-foundations behind both endogenous and exogenous shocks in each context.

\xhdr{Online platforms with rental or reusable products (e.g., cloud computing, car sharing, and hospitality services).}
Rental-based platforms such as AWS (cloud computing), Zipcar (car sharing), and Airbnb (short-term rentals) are natural instances of online assortment with inventory shocks. When a user selects a server instance, vehicle, or home listing, that unit immediately becomes unavailable for a period equal to the rental duration. This leads to an \emph{endogenous shock}, as the temporary depletion of inventory is triggered directly by the platform’s own allocation decision. Once the usage period ends, the same unit returns to the available inventory, replenishing supply. Meanwhile, \emph{exogenous shocks} occur when new resources unexpectedly enter the system---for example, when a new host lists a property on Airbnb, when additional vehicles are relocated to a specific city, or when unused cloud servers are released from internal operations and become available to external customers. The assortment planner must dynamically adjust which items to display to incoming users, accounting for both the predictable restocking of rented items and the uncertain inflow of new supply.

\xhdr{Dynamic resource allocation in public-service and humanitarian operations.}
In social-impact and public-sector contexts---such as refugee resettlement, housing allocation, or hospital bed management---resources are inherently reusable. A resource (e.g., a case manager, a housing unit, or a hospital bed) becomes available again once a prior case or patient exits, which creates an \emph{endogenous shock}. At the same time, \emph{exogenous shocks} occur when additional capacity arrives due to new government funding, emergency shelters, or partner institutions joining the network. 
The central planner’s task is to decide, in real time, which clients or cases to match with available resources given uncertain future arrivals of both clients and capacity. This is effectively an assortment problem over a dynamically fluctuating inventory, directly mapping to our framework with both endogenous and exogenous shocks.

\xhdr{Volunteer nudging and matching platforms.}
Many volunteer coordination platforms (e.g., Catchafire or Food Rescue U.S.) face an online assignment problem where \emph{volunteers} serve as reusable resources and \emph{jobs or requests} arrive sequentially over time. When the platform nudges a volunteer to accept a task, that volunteer becomes temporarily unavailable for the task duration, producing an \emph{endogenous shock}. After completing the task, she becomes available again to accept new opportunities. 
At the same time, \emph{exogenous shocks} occur as new volunteers join, existing volunteers extend their commitment hours, or additional partner organizations contribute workforce capacity. 
The platform’s decision at each arrival corresponds to an online “assortment’’ of volunteers or tasks—deciding which subset of volunteers to nudge for each new job while ensuring balanced utilization and avoiding fatigue. Here, endogenous shocks capture temporary unavailability caused by engagement decisions, whereas exogenous shocks capture new volunteer inflows driven by external outreach or seasonal events. The platform’s goal of maximizing completed jobs subject to dynamic volunteer availability aligns precisely with our model’s structure.

\xhdr{Inventory-sharing and circular-economy systems.}
In sharing-economy or circular supply systems, resources such as tools, clothing, or electronics are repeatedly lent out, repaired, or refurbished, creating \emph{endogenous shocks} that replenish inventory. Each time a borrowed or refurbished unit completes its cycle, it re-enters the pool of available items. In addition, \emph{exogenous shocks} occur through donation campaigns, buy-back programs, or asset transfers from partner organizations. 
For instance, in a shared-equipment platform like Fat Llama\footnote{\url{https://fatllama.com/us}}, a camera lent to a user for three days returns afterward (endogenous), while new listings or donated items are added unpredictably (exogenous). The platform must dynamically decide which items to display to maximize engagement while balancing short-term utilization against uncertain restocking and inflow rates.

\xhdr{Freelance and gig-economy marketplaces.}
In marketplaces such as Upwork, Uber, or DoorDash, workers or service providers act as reusable capacity. When a driver accepts a trip or a freelancer takes on a project, she becomes unavailable for a stochastic service duration—an \emph{endogenous shock} arising directly from the platform’s allocation. After completion, the provider re-enters the pool, replenishing supply. 
Simultaneously, the platform experiences \emph{exogenous shocks} as new drivers sign on, existing workers log back in after a break, or temporary policy changes (e.g., surge incentives, weather effects) alter availability. The platform’s control resembles an assortment decision: which jobs to display or assign to which subset of available providers to optimize earnings and reliability. The interplay between reusable labor capacity (endogenous restocking) and unpredictable supply inflows (exogenous additions) is naturally captured by our model.

\xhdr{Omni-channel and online retail with returns and restocking.}
In omni-channel retailing, inventories “come and go’’ through multiple operational processes. When an online sale is made, the retailer may automatically place a replenishment order with a supplier, which arrives after a lead time---an \emph{endogenous shock} that replenishes the same product based on the previous allocation. Meanwhile, \emph{exogenous shocks} arise when additional units are returned from other sales channels (e.g., physical stores or marketplaces), or when inventory records are corrected due to system discrepancies or shipping cancellations. 
For example, a pair of shoes sold online might be restocked from the distribution center after three days (endogenous), while at the same time another pair may be returned to the warehouse due to a customer refund (exogenous). The retailer’s online assortment (which items to feature or promote) must thus remain robust to both predictable restocking cycles and unpredictable inflows, motivating assortment algorithms that adapt to both types of inventory shocks.\footnote{\revcolor{Notably, lead times are not fixed in practical applications, but for the purpose of theoretical tractability we consider a simplified model with fixed lead times (i.e., fixed duration). We explore the change in durations in our numerical simulations. See \Cref{apx:numerical random instance}. As it turns out, our algorithm is robust to moderate fluctuations in the duration.}}
}

%% file: Paper/apx-numerical.tex
In this section, we provide numerical justifications for the performance of our BIB algorithm on synthetic instances.

\xhdr{Policies.} In our numerical experiments, we compare the revenue of five different policies:
\begin{enumerate}
    \item \textsl{\IBB}~\texttt{(BIB)}: this is the family of policies (\Cref{alg:IBB}) parameterized by penalty function $\pen$ and batch-size threshold $\batchsizethreshold$ proposed in \Cref{sec:IBB}. It generalizes the classic \IB~(IB) when there are inventory shocks. It achieves the asymptotic optimal competitive ratio of $1-\frac{1}{e}\approx 0.632$  under exponential penalty function $\pen(x) = \frac{e^{1-x}-e}{1-e}$ and batch-size threshold $\batchsizethreshold = \sqrt{\mininventory}$ (\Cref{sec:implication}). 
    \item \textsl{\CNUIB}~\texttt{(SCIB)}: this is the family of policies parameterized by penalty function $\pen$. It generalizes the classic IB when there are inventory shocks. Specifically, for each time period $t$ and product $i$, the marginal revenue of all its units is discounted by the same factor of $\pen\left(\min\{\sfrac{\curinventory_{i, t}}{\inventory_i}, 1\}\right)$. Here $\curinventory_{i, t}$ is the inventory level of product $i$ at time period $t$. See \Cref{alg:CNUIB} for its formal description.
    \item \textsl{\CUIB}~\texttt{(DCIB)}: this is the family of policies  parameterized by penalty function $\pen$. It generalizes the classic IB when there are inventory shocks. Specifically, for each time period $t$ and product $i$, the marginal revenue of all its units is discounted by the same factor of $\pen\left(\sfrac{\curinventory_{i, t}}{(\inventory_i+\sum_{\tpre\in[t]}\replen_{i,\tpre}})\right)$. Here $\curinventory_{i, t}$ is the inventory level of product $i$ at time period $t$, and $\replen_{i,t}$ is the new units of product $i$ added due to exogenous shocks. See \Cref{alg:CUIB} for its formal description.
    \item \textsl{\UDIB}~\texttt{(USIB)}: this is the family of policies parameterized by penalty function $\pen$. Each new unit of each product added due to exogenous shocks are treated as a new type of products. This is equivalent to the IBB algorithm (\Cref{alg:IBB}) with batch size threshold $\batchsizethreshold = 1$ (which automatically does not distinguish between ``charging'' and ``ready'' batches as well).
    It generalizes the classic IB when there are inventory shocks.
    \item \textsl{Myopic Greedy}~\texttt{(GREED)}: this is the policy that offers the (myopic) optimal assortment given the available products at each time. This policy is a special case of IB when $\pen(x) = \indicator{x>0}$. Its competitive ratio is at most 0.5.
\end{enumerate}

\input{Paper/alg-CNUIB}
\input{Paper/alg-CUIB}

\subsection{Numerical Performance under Stylized Instances}
\label{apx:numerical stylized instance}
In this section, we first numerically illustrate that three natural modifications of IB --- \CNUIB~(\texttt{SCIB}), \CUIB~(\texttt{DCIB}) and \UDIB~(\texttt{USIB}) --- cannot achieve the optimal asymptotic competitive ratio $(1-\sfrac{1}{e})$ even if we consider online bipartite matching instances (i.e., the special case of online assortment where choice model $\choice(\assortment, i) \in \{0, 1\}$ and $\choice(\assortment, i) = 1$ only if $\assortment = \{i\}$) with large initial inventories and no endogenous inventory shock (i.e., $\duration_i = \infty$ for every product $i$). Then we examine the numerical performance of different policies on these constructed instances.

\subsubsection{Competitive ratio upper bound of 0.552 for \texttt{SCIB}.} 
\label{sec:CR UB SCIB}
Our constructed instance $\matchinginstance(N, \reward, \supply)$ is parameterized by $N \in \naturals$ and $r, s\in[0, 1]$. In this instance, there are $N + 1$ product subgroups $\{\supplies_{\ell}\}_{\ell\in[0:N]}$, indexed by $[0: N]$. Each product subgroup $\supplies_{\ell}$ with $\ell\in[0:N]$ contains $n_\ell$ products, and each product $i$ in subgroup $\supplies_\ell$ has identical per unit price $\reward_i \triangleq \reward^\ell$ and identical initial inventory $\inventory_i \triangleq \inventory$. Here $\inventory$ is a large enough constant. On the other side, there are $2N + 1$ consumer subgroups $\{\demands_k\}_{k\in[0:2N]}$ indexed by $[0: 2N]$. Each consumer subgroup $\demands_k$ with $k\in[0:2N]$ contains $m_k$ consumers, who arrive sequentially.

Consumer subgroup $\demands_0$ has size $m_0 \triangleq \inventory \cdot n_0$. For each $i\in[n_0]$ and each time period $t\in [(i - 1)\inventory + 1:i\inventory]$, a new consumer from subgroup $\demands_0$ who accepts every product from $[n_0 - i + 1] \subseteq \supplies_0$ (i.e., the first $n_0 - i + 1$ products of $\supplies_0$) arrives. 

All remaining consumer subgroups $\{\demands_k\}_{k\in[2N]}$ and products subgroups $\{\supplies_k\}_{\ell\in[N]}$ are constructed recursively based on the execution of \CUIB.\footnote{Since the algorithm is deterministic, it suffices for us to construct instance $\matchinginstance(N, \reward, \supply)$ adaptively, based on the execution of the algorithm.} Specifically, for each $\ell\in[N]$,
\begin{itemize}
    \item Consider time period $t_\ell \triangleq \sum_{k\in[0:2(\ell-1)]}m_k$, i.e., the time period when the first consumer from consumer subgroup $\demands_{2\ell - 1}$ arrives. Let $x_{i, t_\ell}$ be the number of available units of product $i\in\supplies_{\ell - 1}$ at the beginning of time period $t_\ell$ in the algorithm. For each product $i\in\supplies_{\ell - 1}$, there are $\replen_{i,t_\ell} \triangleq \max\{\supply\cdot \inventory - x_{i,t_\ell},0\}$ new units added due to the exogenous inventory shocks. As a sanity check, for each product $i\in\supplies_{\ell-1}$ with $\replen_{i,t_{\ell}} > 0$, its inventory level becomes $\supply\inventory$ and has a discount factor of $\pen(\supply)$ in the algorithm.
    With slight abuse of notation, let $\replen_\ell\triangleq \sum_{i\in\supplies_{\ell - 1}}\replen_{i, t_\ell}$ be the total number of new units added for products from product subgroup $\supplies_{\ell - 1}$. 
    \item Product subgroup $\supplies_{\ell}$ has size $n_{\ell} \triangleq \frac{\replen_\ell}{\inventory\cdot \left(1 - \pen^{-1}\left(\frac{\pen(\supply)}{\reward}\right)\right)}$, where $\pen^{-1}$ is the inverse function of penalty function $\pen$.\footnote{In the construction of instance $\matchinginstance(N, \reward,\supply)$, we impose the requirement that $\frac{\pen(\supply)}{\reward}\in[0, 1]$ for parameters $\reward,\supply\in[0, 1]$. Note that under this requirement, together with the assumption that penalty function $\pen$ is exponential function, the inverse function $\pen^{-1}$ and $n_{\ell}$ are well-defined.} 
    \item Consumer subgroup $\demands_{2\ell - 1}$ have size $m_{2\ell- 1}\triangleq \replen_{\ell}$. For each new unit added for products from product subgroup $\supplies_{\ell-1}$ at time period $t_\ell$ due to exogenous shocks, there is a consumer in consumer subgroup $\demands_{2\ell-1}$ who accepts product $i\in\supplies_{\ell - 1}$ of this unit and one product $i'\in\supplies_{\ell}$. Here product $i'$ is constructed such that for each product $i'\in\supplies_{\ell}$, there are exactly $\frac{\replen_\ell}{n_{\ell}}$ consumers from $\demands_{2\ell - 1}$ accepts this product. Under this construction, for each consumer in $\demands_{2\ell - 1}$, suppose she accepts product $i\in \supplies_{\ell - 1}$ and product $i'\in \supplies_{\ell}$, the algorithm allocates product $i'\in\supplies_{\ell}$ to her and collects price $r^\ell$, since the discounted price of product $i'\in\supplies_{\ell}$ is larger than the discounted price of product $i\in\supplies_{\ell - 1}$ due to the construction of $n_\ell$. More specifically, note that the inventory level of product $i'$ when serving this consumer is at least $\inventory - \frac{\replen_\ell}{n_\ell}$. Thus, the discounted price of product $i'$ is at least 
    \begin{align*}
        \reward_{i'}\pen\left(1 - \frac{\replen_\ell}{\inventory\cdot n_\ell}\right)
        =
        \reward^\ell \pen\left(1 - \frac{\replen_\ell}{\inventory\cdot n_\ell}\right)
        =
        \reward^\ell \frac{\pen(\supply)}{\reward}
        =
        \reward^{\ell - 1}\pen(\supply)
        =
        \reward_{i}\pen(\supply)
    \end{align*}
    where the first and last equalities hold due to the construction of $\reward_{i}, \reward_{i'}$ from subgroups $\supplies_{\ell - 1},\supplies_{\ell}$; and the second equality holds due to the construction of $n_\ell$. Because of this allocation in the algorithm, after serving all consumers from consumer subgroup $\demands_{2\ell - 1}$, the inventory level of each product from product subgroup $\supplies_{\ell}$ is the same, which is equal to $\inventory - \frac{\replen_\ell}{n_\ell}$.
    \item Consumer subgroup $\demands_{2\ell}$ has size $m_{2\ell}\triangleq \inventory \cdot n_\ell$. For the $i$-th consumer in consumer subgroup $\demands_{2\ell}$, she accepts the first $n_\ell - \lfloor\frac{i}{\inventory}\rfloor$ products of product subgroup $\supplies_\ell$.
\end{itemize}
The clairvoyant optimum policy assigns a perfect matching between consumer subgroup $\demands_{2\ell}$ and product subgroup $\supplies_{\ell}$ for every $\ell\in[0:N]$, and assigns every consumer from consumer subgroup $\demands_{2\ell - 1}$ with a unit of her acceptable product from product subgroup $\supplies_{\ell - 1}$ that is added at time period $t_\ell$ due to exogenous shocks for every $\ell\in[N]$. Hence, the total revenue of the clairvoyant optimum policy is $\sum_{\ell\in[0:N]}\reward^\ell\cdot \inventory\cdot n_{\ell} + \sum_{\ell\in[N]}\reward^{\ell - 1}\cdot \replen_\ell$.

On the other end, \texttt{SCIB} is sub-optimal because of the revenue loss from two parts. First, it only assigns an imperfect matching between consumer subgroup $\demands_{2\ell}$ and product subgroup $\supplies_{\ell}$ for every $\ell\in[0:N]$. Moreover, for every $\ell\in[N]$ and every consumer from consumer subgroup $\demands_{2\ell - 1}$, the algorithm allocates a unit of her acceptable product from product subgroup $\supplies_{\ell}$. 

Based on the aforementioned instance construction, we numerically evaluate the performance \texttt{SCIB} with exponential penalty function $\pen(x) = \frac{e^{1-x}-e}{1-e}$ on instance $\matchinginstance(20, 0.5, 0.32)$ with $n_0 = 500$, sufficiently large $\inventory$ and obtain a competitive ratio upper bound of 0.552.

\subsubsection{Competitive ratio upper bound of 0.53 for \texttt{DCIB}.} 
\label{sec:CR UB DCIB}
Our constructed instance $\hat\matchinginstance(N, \reward, \supply)$ is parameterized by $N \in \naturals$ and $r, s\in[0, 1]$. The construction of $\hat\matchinginstance(N, \reward, \supply)$ is almost the same as $\matchinginstance(N, \reward, \supply)$ above, expect that we set $\replen_{i,t_\ell}\triangleq \frac{\supply\cdot \inventory - x_{i,t_\ell}}{1-\supply}$ for every $\ell\in[N]$ and $i\in\supplies_{\ell - 1}$. Under this modification of $\{\replen_{i,t_\ell}\}$, we still ensure that for every $\ell\in[N]$ and every consumer from consumer subgroup $\demands_{2\ell - 1}$, \texttt{DCIB} allocates a unit of her acceptable product from product subgroup $\supplies_{\ell}$, while the clairvoyant optimum policy allocates a unit of her acceptable product from product subgroup $\supplies_{\ell - 1}$.  

Based on the aforementioned instance construction, we numerically evaluate the performance \texttt{DCIB}  with exponential penalty function $\pen(x) = \frac{e^{1-x}-e}{1-e}$ on instance $\hat\matchinginstance(20, 0.5, 0.3)$  with $n_0 = 500$, sufficiently large $\inventory$ and obtain a competitive ratio upper bound of 0.53.

\subsubsection{Competitive ratio upper bound of 0.5 for \texttt{USIB}.} 
\label{sec:CR UB USIB}
Consider an instance $\bar\matchinginstance(\epsilon)$ with $2$ products. Product 1 has price per unit $\reward_1 \triangleq 1$ and a large enough initial inventory $\inventory_1 \triangleq \inventory$. Product 2 has price per unit $\reward_2 \triangleq 1 - \epsilon$ and infinite initial inventory $\inventory_2 \triangleq \infty$. There are $T \triangleq \inventory + 2\inventory^2$ time periods. At each time period $t\in\{\inventory + 1, \inventory + 3, \dots, T - 1\}$, a new unit of product $1$ is added due to exogenous inventory shock. 

For every time period $t \in[\inventory]$, a consumer who only accepts product $1$ arrives. For every time period $t\in\{\inventory + 1, \inventory + 3,\dots, T - 1\}$, a consumer who accepts both products $1$ and $2$ arrives. For every time period $t\in \{\inventory + 2, \inventory + 4, \dots, T\}$, a consumer who only accepts product $1$ arrives.

The clairvoyant optimum policy allocates a unit of product $1$ to every consumer $t\in[t]\cup\{\inventory + 2, \inventory + 4, \dots, T\}$ and collect per unit price $\reward_1 = 1$. Meanwhile, it allocates a unit of product $2$ to every consumer $t\in \{\inventory + 1, \inventory + 3,\dots, T - 1\}$ and collect per unit price $\reward_2 = 1 - \epsilon$. Hence, the total revenue of the clairvoyant optimum policy is $T - \frac{T - \inventory}{2}\epsilon = \inventory + 2\inventory^2 - \inventory^2\epsilon$.

On the other end, \texttt{USIB} allocates a unit of product $1$ to every consumer $t\in[t]\cup\{\inventory + 1, \inventory + 3,\dots, T - 1\}$, and allocates nothing to every consumer $t\in\{\inventory + 2, \inventory + 4, \dots, T\}$. Consequently, the total revenue of the algorithm is $\inventory + \frac{T - \inventory}{2} = \inventory + \inventory^2$.

Putting two pieces together, the competitive ratio of the algorithm on this constructed instance is $\frac{\inventory + \inventory^2}{\inventory + 2\inventory^2 - \inventory^2\epsilon}$ which converges to $\frac{1}{2}$ as $\epsilon$ goes to zero and $\inventory$ goes to infinity.

\subsubsection{Comparison between policies on constructed instances.}
Finally, we examine all policies on the aforementioned instances $\matchinginstance(N, \reward, \supply)$, $\hat\matchinginstance(N, \reward, \supply)$, $\bar\matchinginstance(\epsilon)$, which numerically upper bound the competitive ratio of \texttt{SCIB}, \texttt{DCIB}, \texttt{USIB}, respectively (\Cref{sec:CR UB SCIB,sec:CR UB DCIB,sec:CR UB USIB}). \texttt{BIB}, \texttt{SCIB}, \texttt{DCIB} and \texttt{USIB} all use exponential penalty function $\pen(x) = \frac{e^{1-x}-e}{1-e}$, and the batch-size threshold $\batchsizethreshold$ in \texttt{BIB} is set as $\batchsizethreshold = 10$. 

We record the revenue of different policies in \Cref{tab:numerical results stylized instance}. It can be seen that in all three instances, all four IB policies outperform the myopic greedy policy (\texttt{GREED}). In particular, the revenue gap between \texttt{BIB} and \texttt{GREED} are 15.3\%, 12.0\%, 103.1\% for three instances, respectively. It can also be observed that \texttt{BIB} has the most robust performance for all three instances. In instances $\matchinginstance(10, 0.5, 0.32)$ and $\hat\matchinginstance(10, 0.5, 0.3)$, \texttt{BIB} outperforms \texttt{SCIB} and \texttt{DCIB} significantly, and is close to \texttt{USIB} (which is equivalent to \texttt{BIB} with $\batchsizethreshold = 1$). In instance $\bar\matchinginstance(0.1)$, \texttt{BIB} outperforms \texttt{USIB} significantly, and is close to \texttt{SCIB} and \texttt{DCIB}.  

\begin{table}[ht]
\centering
\caption{Comparing the average revenue of different policies (\Cref{apx:numerical stylized instance}). Instances $\matchinginstance(N, \reward, \supply)$, $\hat\matchinginstance(N, \reward, \supply)$, $\bar\matchinginstance(\epsilon)$ are constructed to upper bound the competitive ratio of \texttt{SCIB}, \texttt{DCIB}, \texttt{USIB}, respectively. For instances $\matchinginstance(N, \reward, \supply)$ and $\hat\matchinginstance(N, \reward, \supply)$, we set $n_0 = 100$ and $\inventory = 50$. For instance $\bar\matchinginstance(\epsilon)$, we set $\inventory = 50$.}
\label{tab:numerical results stylized instance}
\begin{tabular}[t]{lccc}
\toprule
&$\matchinginstance(10, 0.5, 0.32)$
&$\hat\matchinginstance(10, 0.5, 0.3)$
& $\bar\matchinginstance(0.1)$
\\
\midrule
\texttt{BIB}& 5607.72  & 6295.95 & 4570.50\\
\texttt{SCIB}& 5103.57  & 5765.26 & 4800.00\\
\texttt{DCIB}& 5103.57  & 5515.90 & 4800.00\\
\texttt{USIB}& 5848.45& 6658.80 & 2250.00\\
\texttt{GREED} & 4864.29  & 5622.82 & 2250.00\\
\bottomrule
\end{tabular}
\end{table}%

\subsection{Numerical Performance under Random Instances}
\label{apx:numerical random instance}
In this section, we examine the numerical performance of different policies on the following randomly generated instances.

\xhdr{Experimental setup.} We generate the following assortment instances. There are six products indexed by $[6]$ and $T = 3000$ time periods. Each product $i\in[6]$ has an identical initial inventory $\mininventory = 30$ and per unit price $\reward_i$ drawn uniformly between 10 and 25. Without loss of generality, we reindex products such that $\reward_i$ is decreasing in $i\in[6]$. Each product $i\in[6]$ has identical endogenous inventory shock with lead time $\duration_i = T / 3$ and exogenous inventory shock $\replen_{i,t}$ drawn i.i.d.\ from geometric distribution $\texttt{Geo}(0.98)$ for every time period $t\in[T]$.

There are six types of consumers indexed by $1:6$. For each consumer type $j\in[6]$, it is associated with a ``consideration set'' $[j]$ and each product $i\in[j]$ is selected based on a Multinomial Logit (MNL) choice model. Specifically, if assortment $\assortment\subseteq[6]$ is offered, consumers of type $j\in[6]$ selects each product $i\in\assortment$ with probability $\choice^{(j)}(\assortment, i) = \alpha_i^{(j)}/(\alpha_0^{(j)} + \sum_{i'\in\assortment}\alpha_{i'}^{(j)})$, where $\alpha_{i}^{(j)} = 0$ if $i > j$ (i.e., she does not select a product outside of her consideration set), and $\alpha_{i}^{(j)}$ is drawn uniformly between 0.9 and 1.1 if $i\in[j]$. Finally, the outside option (i.e., selecting nothing) is picked by constant probability 0.1 by setting $\alpha_0^{(j)}$ so that $\alpha_0^{(j)}/(\alpha_0^{(j)} + \sum_{i'\in\assortment}\alpha_{i'}^{(j)}) = 0.1$.

In each time period $t\in[T]$, a new consumer is sampled from the following stochastic process similar to \citet{RST-17}. We divide the whole time horizon into 3 equal-length phases, each with 6 equal-length chunks. Thus, there are in total 18 chunks with equal length $\tau = T / 18$. Let $\tau^{(j)} = (6 - j)\tau + 1$ for every $j\in[6]$. At each time period $t\in[T]$, let $t'= t \bmod(T/3)\in[T/3]$, a consumer of type $j\in[6]$ arrives with probability proportional to $e^{-0.001\kappa|t' - \tau^{(j)}|}$. Here, $\kappa$ is a parameter that controls the arrival pattern of the consumer types. Specifically, $\kappa = 0$ corresponds to i.i.d.\ arrival of types from the uniform distribution; and the arrival becomes more heterogeneous with convergence towards a deterministic arrival in descending order of types in every phase as $\kappa$ increases.

We consider four scenarios corresponding to $\kappa \in [4]$. For each scenario, we consider running policies over independent sample paths of the input (including exogenous inventory shocks, type realizations, and consumer choices) by doing 20 iterations of the Monte-Carlo simulation. \IBB, \CNUIB, {\CUIB} and {\UDIB} all use exponential penalty function $\pen(x) = \frac{e^{1-x}-e}{1-e}$, and the batch-size threshold $\batchsizethreshold$ in {\IBB} is set as $\batchsizethreshold = 10$.

\revcolor{In addition to the aforementioned experimental setup, we also consider two further variants, both of which extend beyond the model analyzed in our theoretical results. 
In the first variant, we allow for \emph{negative exogenous inventory shocks}. (Such shocks can be interpreted as unexpected damage to currently available product units.) Specifically, after drawing $\replen_{i,t}$ from the geometric distribution $\texttt{Geo}(0.98)$, we flip its sign (i.e., set $\replen_{i,t} \gets -\replen_{i,t}$) with probability $0.2$.
In the second variant, we allow for \emph{i.i.d.\ stochastic endogenous inventory shocks}. Specifically, for each consumer and her selected product unit, we draw a random duration independently from the geometric distribution $\texttt{Geo}(3/T)$, which has an expected value of $T/3$.
}

\revcolor{
\xhdr{Results.} We record the revenues of all 20 iterations of the Monte-Carlo simulation for each experiment setup. }

For the original setup, we use ``box and whisker'' plots to demonstrate the median revenues and the resulting confidence intervals (\Cref{fig:numerical results random instance}) and compute the average revenue as an estimate for the expected revenue of each policy (\Cref{tab:numerical results random instance}). It is clear from both \Cref{fig:numerical results random instance} and \Cref{tab:numerical results random instance} that all four IB policies outperform the myopic greedy policy and the largest revenue gap is 7.1\%, 7.0\%, 7.8\%, 7.8\% for all four scenarios $\kappa\in[4]$, respectively. Among all four IB policies, {\UDIB} has the lowest revenue, {\CNUIB} has the second lowest revenue. {\IBB} and {\CUIB} have almost the same best revenue.

\begin{figure}
  \centering
       \subfloat[$\kappa = 0$]
      {\includegraphics[width=0.5\textwidth]{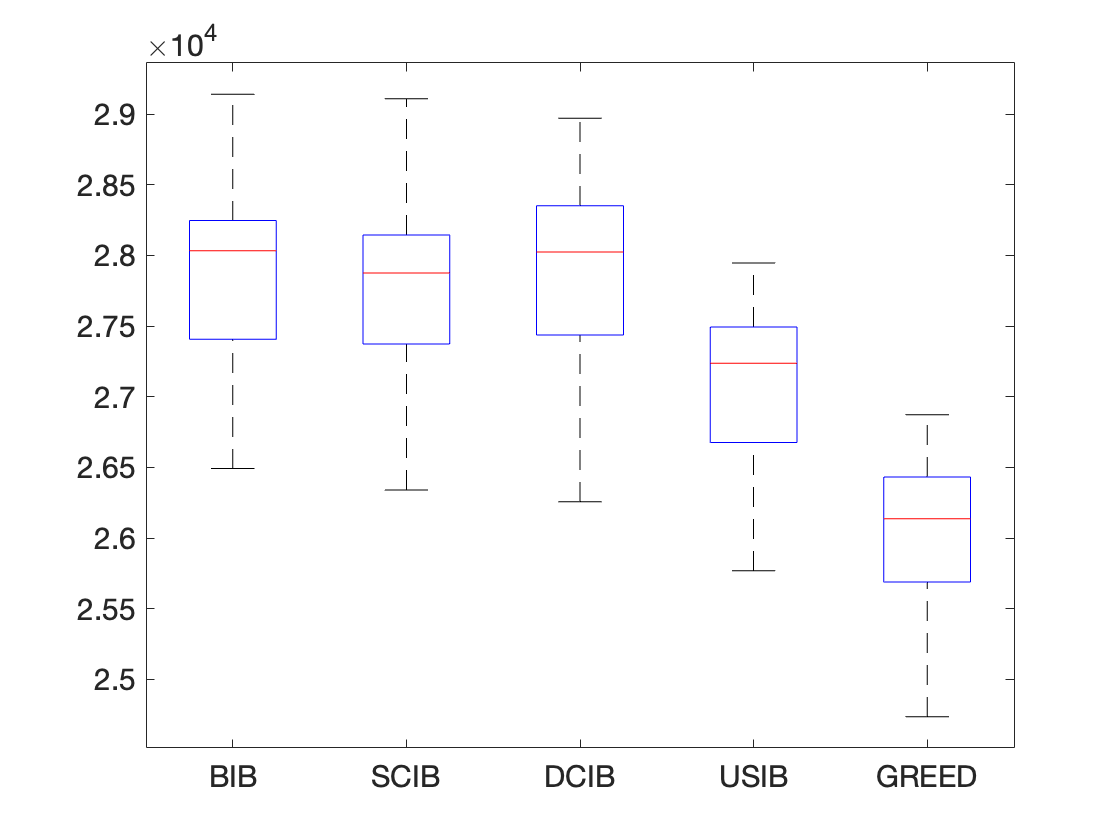}}
      \subfloat[$\kappa = 1$]
      {\includegraphics[width=0.5\textwidth]{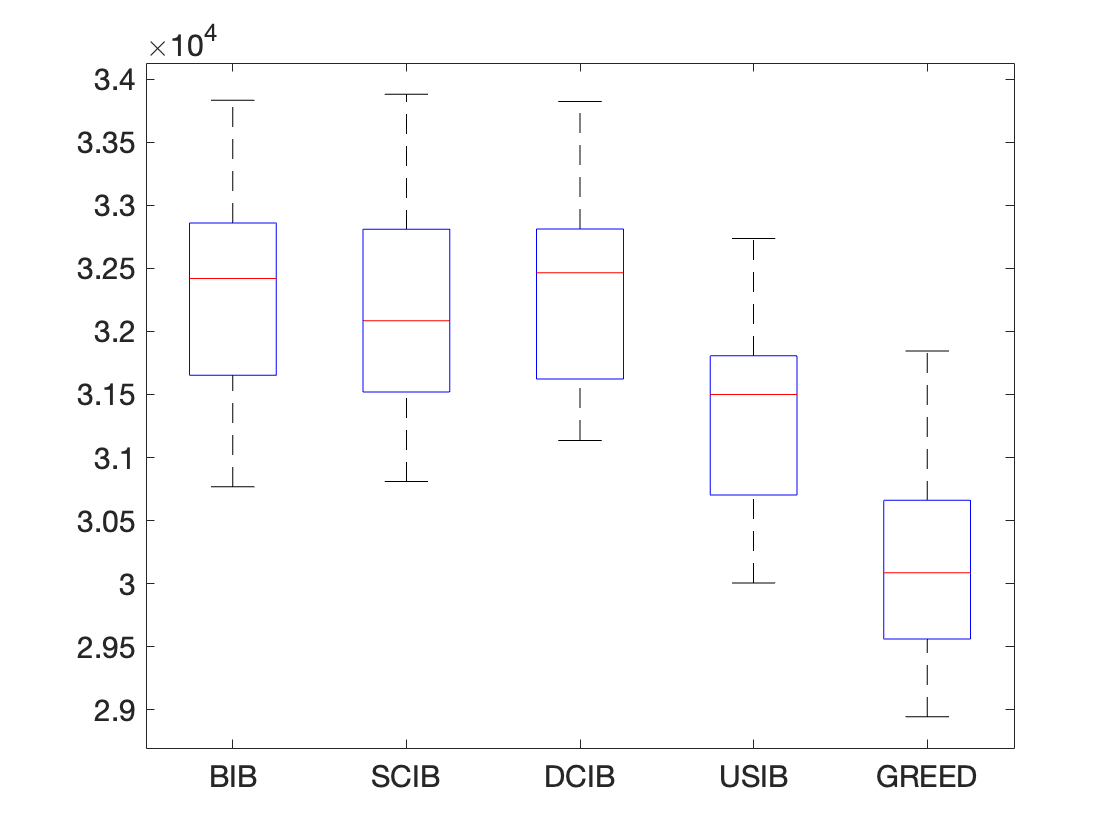}}
      \\
      \subfloat[$\kappa = 2$]
      {\includegraphics[width=0.5\textwidth]{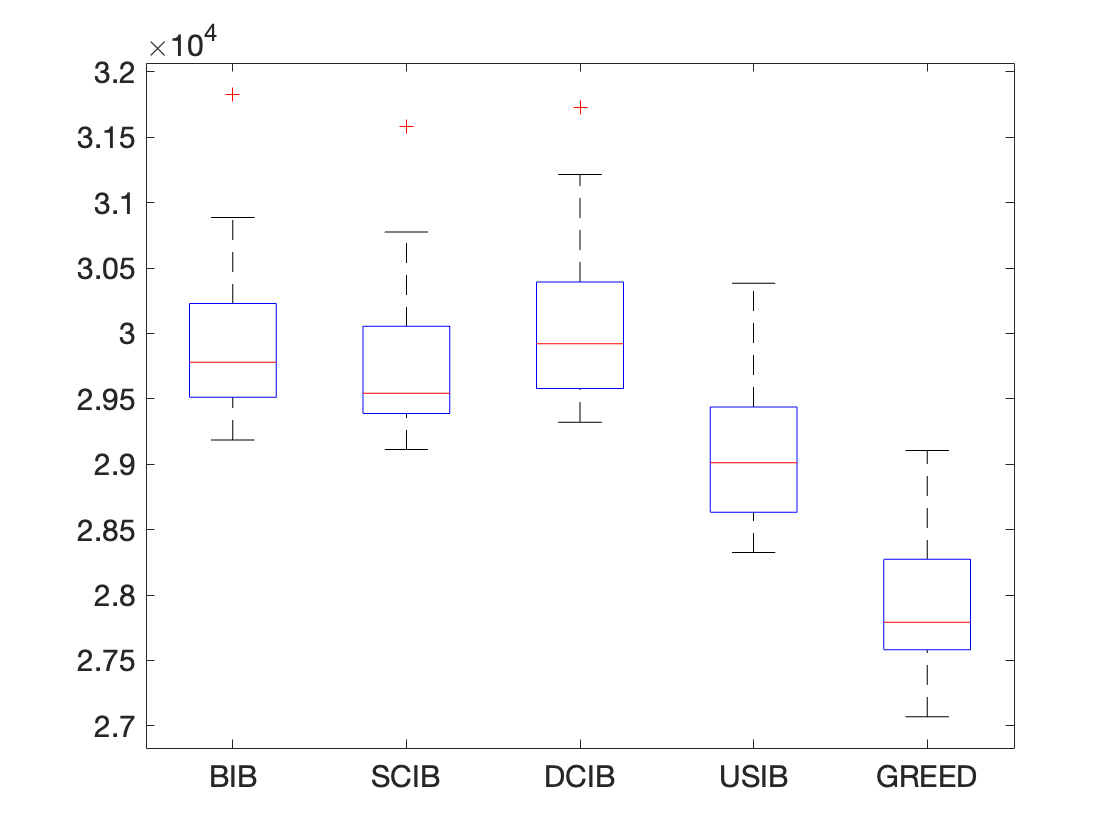}}
      \subfloat[$\kappa = 3$]
      {\includegraphics[width=0.5\textwidth]{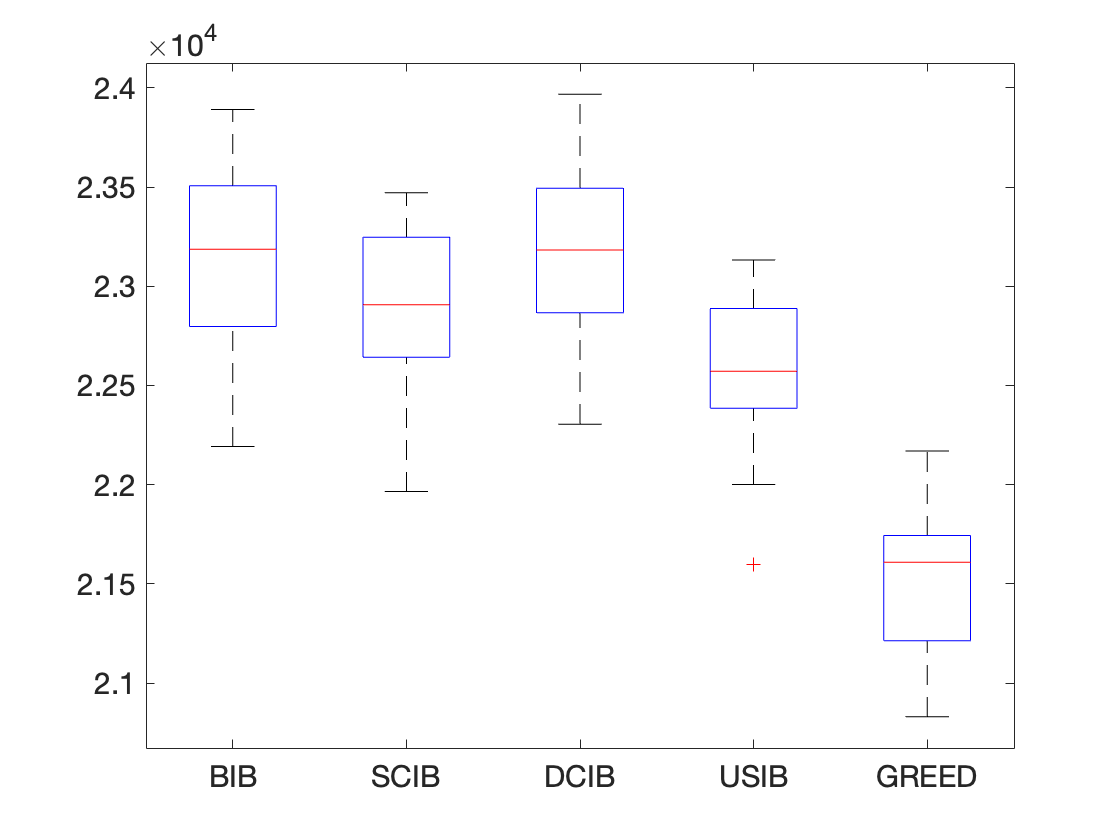}}
  \caption{Box and whisker comparison of revenue under different policies (\Cref{apx:numerical random instance});
  Results are based on 20 iterations of Monte-Carlo simulation.  }
   \label{fig:numerical results random instance}
\end{figure}

\begin{table}
\centering
\caption{\revcolor{Comparing the average revenue of different policies (\Cref{apx:numerical random instance}) in the original setup.}}
\label{tab:numerical results random instance}
\begin{tabular}[t]{lcccc}
\toprule
&$\kappa=0$ &$\kappa=1$&$\kappa=2$&$\kappa=3$\\
\midrule
\texttt{BIB}& 27836.91 & 32308.41 & 29972.79 & 23121.82\\
\texttt{SCIB}& 27747.12 & 32197.73 & 29775.70 & 22914.32\\
\texttt{DCIB}& 27852.76 & 32316.99 & 30101.17 & 23170.69\\
\texttt{USIB}& 27017.10 & 31385.22 & 29070.33 & 22572.46\\
\texttt{GREED} & 25999.49 & 30184.01 & 27926.82 & 21502.78\\
\bottomrule
\end{tabular}
\end{table}%

\revcolor{For the variant with negative exogenous inventory shocks, we compute the average revenue as an estimate of the expected revenue for each policy in \Cref{tab:numerical results random instance negative shocks}. The numerical patterns observed in the original setup largely carry over: all four IB policies consistently outperform the myopic greedy benchmark. However, the revenue gap between the IB policies and the greedy policy becomes smaller. In particular, the largest improvements are approximately $1.9\%$, $1.8\%$, $1.7\%$, and $2.1\%$ for the four scenarios $\kappa \in [4]$, respectively. Among the IB policies, \UDIB{} always yields the lowest revenue, while \CNUIB{} and \CUIB{} perform better and, together with \IBB{}, achieve nearly identical top performance. More specifically, \IBB{} attains the highest revenue when $\kappa=1,2$, \CNUIB{} is best for $\kappa=3$, and \CUIB{} is best for $\kappa=4$.
}

\begin{table}
\centering\revcolor{
\caption{\revcolor{Comparing the average revenue of different policies (\Cref{apx:numerical random instance}) with negative exogenous inventory shocks.}}
\label{tab:numerical results random instance  negative shocks}
\begin{tabular}[t]{lcccc}
\toprule
&$\kappa=0$ &$\kappa=1$&$\kappa=2$&$\kappa=3$\\
\midrule
\texttt{BIB}& 13061.32 & 18746.69 & 20766.04 & 20000.45\\
\texttt{SCIB}& 13056.95 & 18746.49 & 20782.17 & 20002.64\\
\texttt{DCIB}& 13046.72 & 18744.64 & 20765.96 & 20007.51\\
\texttt{USIB}& 12844.23 & 18498.98 & 20542.71 & 19673.44\\
\texttt{GREED}& 12816.43 & 18409.63 & 20435.91 & 19602.96\\
\bottomrule
\end{tabular}
}
\end{table}%

\begin{table}
\centering\revcolor{
\caption{\revcolor{Comparing the average revenue of different policies (\Cref{apx:numerical random instance}) with geometrically distributed stochastic endogenous inventory shocks.}}
\label{tab:numerical results random instance random shocks}
\begin{tabular}[t]{lcccc}
\toprule
&$\kappa=0$ &$\kappa=1$&$\kappa=2$&$\kappa=3$\\
\midrule
\texttt{BIB}& 13061.32 & 18746.69 & 20766.04 & 20000.45\\
\texttt{SCIB}& 13056.95 & 18746.49 & 20782.17 & 20002.64\\
\texttt{DCIB}& 13046.72 & 18744.64 & 20765.96 & 20007.51\\
\texttt{USIB}& 12844.23 & 18498.98 & 20542.71 & 19673.44\\
\texttt{GREED}& 12816.43 & 18409.63 & 20435.91 & 19602.96\\
\bottomrule
\end{tabular}
}
\end{table}%

\revcolor{For the variant with geometrically distributed endogenous inventory shocks, the average revenues of all policies are reported in \Cref{tab:numerical results random instance random shocks}. Consistent with the previous setups, all four IB policies outperform the myopic greedy benchmark. The largest revenue improvements are about $6.3\%$, $6.0\%$, $5.5\%$, and $4.7\%$ for the four scenarios $\kappa \in [4]$, respectively. Among the IB policies, {\UDIB} again yields the lowest revenue, while {\IBB}, {\CNUIB}, and {\CUIB} remain closely matched and achieve highly competitive performance across all scenarios, with only minor variations in which policy attains the highest revenue.}

%% file: Paper/alg-CNUIB.tex
\begin{algorithm}
 	\caption{\textsc{\CNUIB}}
 	\label{alg:CNUIB}
	\KwIn{penalty function $\pen$}

 	\For{each time period $t\in[T]$}{ 
 	\tcc{consumer $t$ with choice model $\choice_t$ arrives.} 
 	
    \tcc{$\repleni$ new units of product $i$ are added due to exogenous shock.}
    
    update inventory level $\{\curinventory_{i,t}\}_{i\in[n]}$ of all products at time period $t$.

    solve $\assortment_t\gets
    \argmax\nolimits_{
    \assortment\in \assortmentspace}
    \sum_{i\in \assortment} 
    \reward_i
    \pen\left(\min\{\sfrac{\curinventory_{i, t}}
    {\inventory_i}, 1\}
    \right)\cdot \choice_{t}(\assortment, i)$.
 	
 	allocate assortment $\assortment_t$ to consumer $t$.
 	}
\end{algorithm}

%% file: Paper/alg-CUIB.tex
\begin{algorithm}
 	\caption{\textsc{\CUIB}}
 	\label{alg:CUIB}
	\KwIn{penalty function $\pen$}

 	\For{each time period $t\in[T]$}{ 
 	\tcc{consumer $t$ with choice model $\choice_t$ arrives.} 
 	
    \tcc{$\repleni$ new units of product $i$ are added due to exogenous shock.}
    
    update inventory level $\{\curinventory_{i,t}\}_{i\in[n]}$ of all products at time period $t$.

    solve $\assortment_t\gets
    \argmax\nolimits_{
    \assortment\in \assortmentspace}
    \sum_{i\in \assortment} 
    \reward_i
    \pen\left(\sfrac{\curinventory_{i, t}}
    {(\inventory_i+\sum_{\tpre\in[t]}\replen_{i,\tpre}}
    )\right)\cdot \choice_{t}(\assortment, i)$.
 	
 	allocate assortment $\assortment_t$ to consumer $t$.
 	}
\end{algorithm}

%% file: Paper/apx-IAP-illustration.tex
\revcolor{
In this section, we present instances of the interval assignment problem (IAP) to illuminate its combinatorial structure, motivate the need for a carefully constructed assignment, and the execution of our proposed \Cref{alg:interval assignment}.

\xhdr{The failure of real consumption levels.}
As explained in \Cref{sec:IBB proof} and in the proof of \Cref{thm:concave competitive ratio}, the output of the IAP is used to define a transformed version of the consumption level (i.e., total initial inventory minus the remaining inventory at the time of allocation). This raises a natural question: does the real consumption level—defined as $\containni = |\{j : a_i \in [a_j, b_j]\}|$ for each interval $i$—already satisfy the three defining properties of the IAP, namely \hyperref[lem:IAP]{\LocalDominance}, \hyperref[lem:IAP]{\GlobalDominance}, and \hyperref[lem:IAP]{\PartitionMonotonicity}? 
The answer is negative. The following example demonstrates that the naive assignment $\assign_i = \containni$ may violate \hyperref[lem:IAP]{\PartitionMonotonicity}, even though it satisfies the other two properties by construction.

\begin{example}[A Path of Overlapping Intervals]
\label{example:failure of actual consumption level}
Consider four intervals arranged in a ``path'' where each overlaps only with its immediate predecessor and successor:
\begin{align*}
    [a_1, b_1] = [1, 4], \quad
    [a_2, b_2] = [3, 6], \quad
    [a_3, b_3] = [5, 8], \quad
    [a_4, b_4] = [7, 10].
\end{align*}
Here, the left endpoints are strictly increasing ($a_1 < a_2 < a_3 < a_4$), and for each $i$, the coverage set $\contain_i = \{j : a_i \in [a_j, b_j]\}$ satisfies $\contain_1 = \{1\}$, $\contain_2 = \{1,2\}$, $\contain_3 = \{2,3\}$, and $\contain_4 = \{3,4\}$. An illustration is provided in \Cref{fig:failure of actual consumption level}.
\end{example}

\begin{figure}
\centering
\input{figs/fig-IAP-failure-of-real-consumption}
\caption{\revcolor{\Cref{example:failure of actual consumption level}: A path of four overlapping intervals. Each interval overlaps only with its immediate neighbor. Both endpoints $(a_i, b_i)$ for each interval $i$ are marked in red. The bold number inside each interval indicates its IAP assignment $\assign_i$, and the blue dashed arrows illustrate the predecessor links $\pfather_j = i$ that define the partition satisfying \hyperref[lem:IAP]{\PartitionMonotonicity}.}}
\label{fig:failure of actual consumption level}
\end{figure}
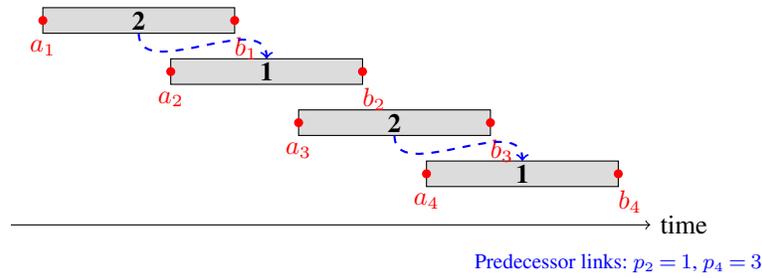

Under the naive assignment $\assign_i = \containni$, we obtain $(\assign_1, \assign_2, \assign_3, \assign_4) = (1, 2, 2, 2)$ for \Cref{example:failure of actual consumption level}. This assignment satisfies \hyperref[lem:IAP]{\LocalDominance} (each coverage set contains labels that dominate $\{1,\dots,\containni\}$) and \hyperref[lem:IAP]{\GlobalDominance} (since $\assign_i = \containni$, the inequality holds with equality for any decreasing concave penalty). However, it \emph{violates} \hyperref[lem:IAP]{\PartitionMonotonicity}: there is only one interval labeled ``1'', yet three intervals labeled ``2''. Consequently, it is impossible to partition the intervals into groups such that each group of size $m$ contains exactly the labels $\{1, 2, \dots, m\}$. 

This failure is significant because \hyperref[lem:IAP]{\PartitionMonotonicity} plays a crucial role in aligning the primal and dual objective values in the proof of \Cref{thm:concave competitive ratio}. Without it, setting the dual variables equal to the real consumption levels would not recover the tight approximation guarantee established by \citet{GNR-14} in the special case with no endogenous shock. 

In contrast, \Cref{alg:interval assignment} produces a valid IAP assignment for \Cref{example:failure of actual consumption level}. Specifically, it outputs $(\assign_1, \assign_2, \assign_3, \assign_4) = (2, 1, 2, 1)$, which satisfies all three properties. It constructs this assignment in the following order: $\assign_3 \gets 2$, $\assign_1 \gets 2$, $\assign_4 \gets 1$, and $\assign_2 \gets 1$. As illustrated in \Cref{fig:failure of actual consumption level}, the resulting labels enable a partition into two disjoint chains—$\{1,2\}$ and $\{3,4\}$—each of which satisfies \hyperref[lem:IAP]{\PartitionMonotonicity}.

\xhdr{The failure of a simple greedy procedure.}
\Cref{example:failure of actual consumption level} demonstrates that the naive assignment based on real consumption levels is insufficient, thereby motivating the need for an algorithm that explicitly enforces all three properties of IAP. 

A natural follow-up question is whether a simpler algorithm—beyond the full IAP construction—might suffice. One such candidate is the following \emph{greedy} procedure: process the intervals in increasing order of their left endpoints, and for each interval, assign the \emph{smallest} positive integer that maintains \hyperref[lem:IAP]{\LocalDominance} with respect to all previously processed coverage sets. The rationale is that \hyperref[lem:IAP]{\LocalDominance} only constrains intervals whose left endpoints have already been encountered, making a forward greedy choice seemingly safe.

Indeed, for \Cref{example:failure of actual consumption level}, this greedy procedure yields the assignment $(1, 2, 1, 2)$, which satisfies \hyperref[lem:IAP]{\LocalDominance}, \hyperref[lem:IAP]{\GlobalDominance}, and \hyperref[lem:IAP]{\PartitionMonotonicity}. However, as the next example shows, this procedure does not generalize: on more complex instances, the greedy assignment may satisfy \hyperref[lem:IAP]{\LocalDominance} but violate \hyperref[lem:IAP]{\PartitionMonotonicity}, and thus fail to meet the full requirements of the IAP.

\begin{example}[A Binary-Tree Structure of Intervals]
\label{example:failure of greedy}
Consider seven intervals arranged in a complete binary tree of depth 3 (i.e., three levels), defined as follows:
\begin{align*}
\text{Level 1 (root):} \quad & [a_1, b_1] = [1, 14], \\
\text{Level 2 (children of root):} \quad & [a_2, b_2] = [2, 7], \quad [a_5, b_5] = [8, 13], \\
\text{Level 3 (leaves):} \quad & [a_3, b_3] = [3, 4], \quad [a_4, b_4] = [5, 6], \quad
[a_6, b_6] = [9, 10], \quad [a_7, b_7] = [11, 12].
\end{align*}
where the left endpoints are strictly increasing ($a_1 < a_2 < \dots < a_6 < a_7$).  
Each interval is fully contained within its parent, and for every $i$, the coverage set $\contain_i = \{j : a_i \in [a_j, b_j]\}$ consists precisely of interval $i$ and all its ancestors in the tree. An illustration is provided in \Cref{fig:failure of greedy}.
\end{example}

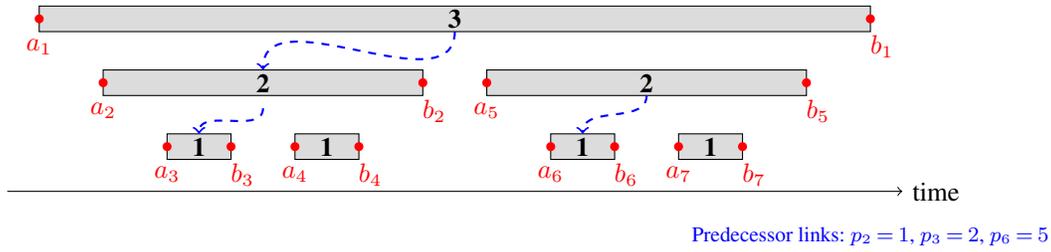
\begin{figure}
\centering
\input{figs/fig-IAP-failure-of-greedy}
\caption{\revcolor{\Cref{example:failure of greedy}: A complete binary-tree hierarchy of seven intervals. Each parent fully contains its two children. Both endpoints $(a_i, b_i)$ for each interval $i$ are marked in red. The bold number inside each interval indicates its IAP assignment $\assign_i$, and the blue dashed arrows illustrate the predecessor links $\pfather_j = i$ that define the partition satisfying \hyperref[lem:IAP]{\PartitionMonotonicity}.}}
\label{fig:failure of greedy}
\end{figure}

For \Cref{example:failure of greedy}, the aforementioned greedy procedure yields the assignment $(1, 2, 3, 3, 2, 3, 3)$,\footnote{The real consumption levels induce the same assignment, which violates \hyperref[lem:IAP]{\PartitionMonotonicity}, for \Cref{example:failure of greedy}.} which satisfies \hyperref[lem:IAP]{\LocalDominance}, \hyperref[lem:IAP]{\GlobalDominance}, but violates \hyperref[lem:IAP]{\PartitionMonotonicity}: there is only one interval labeled ``1'', yet two intervals labeled ``2'', and four intervals labeled ``3''. Consequently, it is impossible to partition the intervals into groups such that each group of size $m$ contains exactly the labels $\{1, 2, \dots, m\}$. 

In contrast, although more involved, \Cref{alg:interval assignment} produces a valid IAP assignment for \Cref{example:failure of greedy}. Specifically, it outputs assignment $(3, 2, 1, 1, 2, 1, 1)$, which satisfies all three properties.\footnote{\revcolor{In \Cref{example:failure of greedy}, the valid IAP assignment $(3, 2, 1, 1, 2, 1, 1)$ can be considered as the inventory level. However, this does not imply that the real inventory level is always valid for IAP assignment. 
To see this, consider a modification of \Cref{example:failure of greedy} in which the left endpoints are reordered—specifically, set $a_3 = 1$ and $a_1 = 3$, while keeping all other left and right endpoints unchanged. In this modified example, the inventory level becomes $(1, 2, 3, 1, 2, 1, 1)$. Although this assignment satisfies \hyperref[lem:IAP]{\GlobalDominance} and \hyperref[lem:IAP]{\PartitionMonotonicity}, it \emph{violates} \hyperref[lem:IAP]{\LocalDominance}: for interval $4$, the coverage set $\contain_4$ includes intervals whose assigned values are insufficient to dominate the required set $\{1, \dots, |\contain_4|\}$.}} It constructs this assignment in the following order: $\assign_1 \gets 3$, $\assign_5 \gets 2$, $\assign_2 \gets 2$, $\assign_7 \gets 1$, $\assign_6 \gets 1$, $\assign_4 \gets 1$, and $\assign_3 \gets 1$. As illustrated in \Cref{fig:failure of greedy}, the resulting labels enable a partition into disjoint chains—$\{1,2,3\}$, $\{5,6\}$, $\{4\}$, and $\{7\}$—each of which satisfies \hyperref[lem:IAP]{\PartitionMonotonicity}.

}

%% file: figs/fig-IAP-failure-of-real-consumption.tex
\begin{tikzpicture}[scale=0.85, font=\small]
  \definecolor{intervalcolor}{RGB}{220,220,220} 

  \filldraw[intervalcolor, draw=black] (1,3) rectangle (4,3.4);
  \node at (2.5,3.2) {\textbf{2}}; 

  \filldraw[intervalcolor, draw=black] (3,2.2) rectangle (6,2.6);
  \node at (4.5,2.4) {\textbf{1}}; 

  \filldraw[intervalcolor, draw=black] (5,1.4) rectangle (8,1.8);
  \node at (6.5,1.6) {\textbf{2}}; 

  \filldraw[intervalcolor, draw=black] (7,0.6) rectangle (10,1.0);
  \node at (8.5,0.8) {\textbf{1}}; 

  \foreach \x/\y/\i in {
    1/3.2/1,
    3/2.4/2,
    5/1.6/3,
    7/0.8/4
  } {
    \fill[red] (\x,\y) circle (2pt);
    \node[below, red] at (\x,\y-0.15) {$a_{\i}$};
  }

  \foreach \x/\y/\i in {
    4/3.2/1,
    6/2.4/2,
    8/1.6/3,
    10/0.8/4
  } {
    \fill[red] (\x,\y) circle (2pt);
    \node[below, red] at (\x+0.18,\y-0.1) {$b_{\i}$};
  }

  \draw[->] (0.5,0) -- (10.5,0) node[right] {time};

  \tikzset{chainarrow/.style={->, dashed, blue, thick}}

  \draw[chainarrow] (2.5,3.0) to[out=-90,in=90] (4.5,2.6);

  \draw[chainarrow] (6.5,1.4) to[out=-90,in=90] (8.5,1.0);

  \node[blue, font=\scriptsize] at (10,-0.6) {Predecessor links: $\pfather_2 = 1$, $\pfather_4 = 3$};
\end{tikzpicture}

%% file: figs/fig-IAP-failure-of-greedy.tex
\begin{tikzpicture}[scale=0.85, font=\small]
  \definecolor{intervalcolor}{RGB}{220,220,220} 

  \filldraw[intervalcolor, draw=black] (1,4) rectangle (14,4.4);
  \node at (7.5,4.2) {\textbf{3}}; 

  \filldraw[intervalcolor, draw=black] (2,3) rectangle (7,3.4);
  \node at (4.5,3.2) {\textbf{2}}; 
  \filldraw[intervalcolor, draw=black] (8,3) rectangle (13,3.4);
  \node at (10.5,3.2) {\textbf{2}}; 

  \filldraw[intervalcolor, draw=black] (3,2) rectangle (4,2.4);
  \node at (3.5,2.2) {\textbf{1}}; 
  \filldraw[intervalcolor, draw=black] (5,2) rectangle (6,2.4);
  \node at (5.5,2.2) {\textbf{1}}; 
  \filldraw[intervalcolor, draw=black] (9,2) rectangle (10,2.4);
  \node at (9.5,2.2) {\textbf{1}}; 
  \filldraw[intervalcolor, draw=black] (11,2) rectangle (12,2.4);
  \node at (11.5,2.2) {\textbf{1}}; 

  \foreach \x/\y/\i in {
    1/4.2/1,
    2/3.2/2,
    3/2.2/3,
    5/2.2/4,
    8/3.2/5,
    9/2.2/6,
    11/2.2/7
  } {
    \fill[red] (\x,\y) circle (2pt);
    \node[below, red] at (\x,\y-0.15) {$a_{\i}$};
  }

  \foreach \x/\y/\i in {
    14/4.2/1,
    7/3.2/2,
    4/2.2/3,
    6/2.2/4,
    13/3.2/5,
    10/2.2/6,
    12/2.2/7
  } {
    \fill[red] (\x,\y) circle (2pt);
    \node[below, red] at (\x+0.18,\y-0.1) {$b_{\i}$};
  }

  \draw[->] (0.5,1.5) -- (14.5,1.5) node[right] {time};

  \tikzset{chainarrow/.style={->, dashed, blue, thick}}

  \draw[chainarrow] (7.5,4.0) to[out=-90,in=90] (4.5,3.4);
  \draw[chainarrow] (4.5,2.8) to[out=-90,in=90] (3.5,2.4);

  \draw[chainarrow] (10.5,3.0) to[out=-90,in=90] (9.5,2.4);


  \node[blue, font=\scriptsize] at (14,0.8) {Predecessor links: $\pfather_2=1$, $\pfather_3=2$, $\pfather_6=5$};
\end{tikzpicture}

%% file: Paper/apx-open-problem.tex
\revcolor{
Our analysis of online assortment under inventory shocks is built on a clean abstraction of reusability---where endogenous shocks model the automatic return of allocated units after a fixed usage duration---and exogenous shocks capture unpredictable inflows of new supply. While this model captures a wide range of applications, several natural extensions arise from both practical considerations and theoretical curiosity. We outline four promising directions for future work.

\xhdr{Negative inventory shocks and loss modeling.}  
Our current model assumes non-negative exogenous shocks. However, in real-world settings, inventory may be unexpectedly reduced due to damage, spoilage, theft, or system errors---phenomena we model as \emph{negative exogenous shocks}. As shown in \Cref{apx:numerical}, BIB exhibits strong empirical performance even when negative shocks are introduced. From a theoretical perspective, negative shocks introduce nonlinearity into the inventory dynamics (e.g., the ``available inventory'' is the max of zero and prior inventory minus shock), which complicates the design of a tractable LP benchmark. Developing competitive algorithms with formal guarantees under such adverse shocks---potentially by combining our batching ideas with robustness techniques from convex optimization and duality analysis---is a compelling direction for future research.\footnote{\revcolor{Using a standard charging argument, the myopic greedy policy (equivalently, the Inventory Balancing algorithm with penalty function $\pen(x) = \indicator{x > 0}$) can be shown to be $0.5$-competitive under adversarial negative shocks.}}

\xhdr{Stochastic and heterogeneous return durations.}  
Our analysis assumes deterministic and homogeneous usage durations $d_i$. However, in many applications (e.g., gig platforms, hospital beds, rental markets), the duration a unit remains unavailable is random and may vary across allocations. Recent work by \citet{KLU-25} shows that our batched framework can be lifted via a black-box reduction to settings with i.i.d.\ stochastic return times, preserving the $1-1/e$ competitive ratio under mild conditions. An open question is whether similar guarantees can be obtained for adversarial or highly heterogeneous duration distributions, or whether new algorithmic ideas are required.

\xhdr{Endogenous shocks from operational replenishment policies.}  
In omni-channel retail and supply chain settings, endogenous inventory replenishment often stems not from usage-driven returns but from operational policies such as $(s,S)$-type or threshold-based restocking. In such models, a replenishment order is triggered when on-hand inventory falls below a threshold, and units arrive after a (possibly stochastic) lead time. Extending our framework to accommodate such endogenous dynamics would require reinterpreting the ``return duration'' as a lead time and incorporating the replenishment decision into the adversary’s power or the algorithm’s control. 

\xhdr{Incorporating inventory holding and restocking costs.}  
Our objective focuses solely on revenue maximization, ignoring potential costs associated with holding inventory or placing restocking orders. Extending the model to include holding costs (e.g., per-unit-per-period) or fixed restocking costs would shift the objective toward profit maximization and likely favor more conservative allocation policies. Integrating such costs into the primal-dual analysis---perhaps via modified penalty functions or cost-aware dual variables---remains an intriguing challenge.

\xhdr{Joint inventory-assortment control under shocks.}  
In our model, exogenous shocks are fully adversarial and uncontrollable. However, in practice, platforms may exert partial control over inventory inflows---for example, by procuring additional units, triggering returns, or managing volunteer sign-up campaigns. A natural generalization is a \emph{joint inventory-assortment control} problem, where the decision maker simultaneously chooses both the assortment and an inventory action (e.g., order quantity, promotional effort) at each time step. Developing competitive algorithms for this richer setting would require coupling the batched inventory balancing framework with dynamic inventory control techniques.
}